\documentclass[smallextended]{svjour3}

\usepackage[margin=1in]{geometry}
\usepackage{latexsym,graphicx,amssymb,amsmath}
\usepackage[dvipsnames,usenames]{color}

\newcommand{\RR}{{\mathbb R}}

\newcommand{\T}{{\mathcal T}}

\newcommand{\diag}{{\operatorname{diag}}}
\newcommand{\LSA}{{\operatorname{LSA}}}

\newcommand{\M}{\mathcal M}

\definecolor{OliveGreen}{rgb}{.3,.5,.2}
\definecolor{MidnightBlue}{rgb}{.3,.4,.6}
\usepackage[utf8]{inputenc}
 \definecolor{codegreen}{rgb}{0,0.6,0}
\definecolor{codegray}{rgb}{0.5,0.5,0.5}
\definecolor{codepurple}{rgb}{0.58,0,0.82}
\definecolor{backcolour}{rgb}{0.95,0.95,0.92}
 \newcommand{\tc}{\text{:}}

\smartqed

\title{Identifiability of species network topologies from genomic sequences using the logDet distance} 

\author{Elizabeth S. Allman \and Hector Ba\~nos \and John A. Rhodes }

\institute{ Elizabeth S. Allman \and John A. Rhodes \at Department of Mathematics and Statistics, University of Alaska Fairbanks, Fairbanks, AK, 99775, USA \\
\email{e.allman@alaska.edu, j.rhodes@alaska.edu}
\and
Hector D. Ba\~nos \at Department of Biochemistry \& Molecular Biology, Faculty of Medicine, Dalhousie University, Halifax, Nova Scotia, CANADA \\ \email{hbanos@dal.ca}}

\begin{document}

\maketitle
	
\begin{abstract} 
Inference of network-like evolutionary relationships between species from genomic data must address the interwoven signals from both gene flow and incomplete lineage sorting. The heavy computational demands of standard approaches to this problem severely limit the size of datasets that may be analyzed, in both the number of species and the number of genetic loci.
Here we provide a theoretical pointer to more efficient methods, by showing that logDet distances computed from genomic-scale sequences retain sufficient information to recover network relationships in the level-1 ultrametric case. This result is obtained under the Network Multispecies Coalescent model
combined with a mixture of General Time-Reversible sequence evolution models across individual gene trees, but does not depend on partitioning sequences by genes.
Thus  under standard stochastic models statistically justifiable inference of network relationships from sequences can be accomplished without consideration of individual genes or gene trees.
\end{abstract}

\maketitle

\section{Introduction}

As genomic-scale sequencing has become increasingly common, attention in phylogenetics has shifted from inferring trees of evolutionary relationships for individual genetic loci from a set of species to inferring relationships between the species themselves. A substantial complication is that population genetic processes within species, as modeled by the \emph{Multispecies Coalescent} (MSC) model can lead  to individual gene trees having quite different topological structures than the tree relating the species overall. If the evolutionary history of the species also involved hybridization or other forms of horizontal gene flow, so that a species network is a more suitable depiction of relationships, the relationships of gene trees to the network, as modeled by the \emph{Network Multispecies Coalescent} (NMSC) model, is even more complex.

Inference of species networks, through a combined NMSC and sequence substitution model, can be performed in a Bayesian framework \cite{Zhang2017,Wen2018} but computational demands severely limit both the number of taxa and the number of genetic loci considered. Other methods take a faster two-stage approach, first inferring gene trees which are treated as ``data" for a second inference of a species network. Approaches include maximum pseudolikelihood  using either rooted triples (PhyloNet) or quartets (SNaQ) displayed on the gene trees \cite{Yu2015,Solis-Lemus2016}, or the faster,  distance-based analysis built on gene quartets of NANUQ \cite{ABR2019}. Still, the first stage of these approaches, the inference of individual gene trees, can be a major computational burden. Avoiding such gene tree inference, and passing more directly from sequences to an inferred network, could substantially reduce total computational time in data analysis pipelines.

The goal of this paper is to show that most topological features of a level-1 species network can be identified from logDet intertaxon distances computed from aligned genomic-scale sequences. In particular this can be done without partitioning the sequences by genes, under a combined model of the NMSC and a mixture of general time-reversible (GTR) substitution processes on gene trees. While the main result, that the logDet distances retain enough information to recover most of the species network, despite having lost information on individual genes, is a theoretical one, it points the way toward faster algorithms for practical inference. In particular, since the computation of logDet distances requires little effort, it suggests that a distance-based approach similar to NANUQ's, but avoiding individual gene tree inference, may offer substantially faster analyses than current methods. 

\medskip

The model of sequence evolution underlying our result accounts not only for base substitutions along each gene tree, but also for variation in gene trees due to their formation under a coalescent process combined with hybridization or similar gene transfer.  Our model extends to networks the mixture of coalescent mixtures model on species trees of \cite{Allman2019}, 
which itself extended the coalescent mixture introduced by Chifman and Kubatko \cite{Chifman2015}.  More specifically, for a fixed species network, gene trees are  formed under the Network Multispecies Coalescent model \cite{Meng2009,Yu2011,Zhu2016} for each site independently. GTR substitution parameters for base evolution on each site's tree are then independently chosen from some distribution, leading to a site pattern distribution. These site distributions are finally combined to give a 
site pattern distribution for genomic sequences. (As discussed in Section \ref{sec::networksmodels}, this distribution also applies to a more realistic model in which multisite genes with a single substitution process have lengths chosen independently from some distribution.)
While this pattern frequency distribution thus reflects the substitution processes on all the gene trees,  information about pattern frequencies arising on any individual gene tree is hidden.

The logDet distance was first introduced  in the context of a single class general Markov model of sequence evolution on a single gene tree \cite{Steel94,Lockhart94}, and has been used both to obtain both gene tree identifiability results and for inference of individual gene trees.
Considering genomic sequences,  Liu and Edwards \cite{Liu2009}, and independently
Dasarathy et al. \cite{Roch2015}, showed that for a Jukes-Cantor substitution model and an ultrametric species tree, the Jukes-Cantor distances obtained under the coalescent mixture model still allowed for consistent inference of topological species trees. By passing to the logDet distance, Allman et al. \cite{Allman2019} extended this result to the more realistic mixture of coalescent mixtures model, showing that the logDet distance allowed for consistent inference of a topological species tree, assuming it is ultrametric in generations. This study builds on all these works on gene and species tree models,  but considers level-1 species networks on which all extant species are equidistant from the root.

Passing from species trees to networks is a substantial step, however, and our approach is strongly motivated by the approach taken by Ba\~nos \cite{Banos2019} in studying identifiability of features of unrooted level-1 topological species networks from gene tree quartet concordance factors (probabilities of the different quartet topologies displayed on gene trees). In the ultrametric setting of this work, we show that logDet distances computed from genomic sequences suffice to determine 4-cycles on undirected rooted triple networks, and then that this 4-cycle information for different rooted triples can be combined to determine all cycles of size 4 or more, and even all hybrid nodes in those cycles of size 5 or more. We do not obtain information on 2- or 3- cycles, so our results closely parallel those in \cite{Banos2019}, despite the rather different source of information.

\medskip
There are a number of other theoretical works in the literature on determining phylogenetic networks from limited information. For instance, \cite{Jansson2006} investigates determining a level-1 network from the rooted triple trees it displays,
\cite{HuberEtAl2017,HuberEtAl2018} discuss how knowledge of trinets (induced 3-taxon directed rooted networks) and quarnets (induced 4-taxon undirected unrooted networks) determine larger networks, and \cite{VanIersaletal2020} explores determination of networks from distances. However, the question of how, or whether, these results can be applied to biological data is not addressed, and the setting of these works is not directly applicable to obtaining our results. 

Other works \cite{GrossLong18,GrossEtAl2020,HolleringSullivant2021} use algebraic approaches to show that certain types of level-1 networks can be identified from joint pattern frequency arrays under group-based models of sequence evolution such as the Jukes-Cantor and Kimura models. In addition to their restriction on sequence evolution models, these works do not incorporate a coalescent process. That is, all sequence sites are assumed to have evolved on one of the finitely-many trees displayed on the network. Since the absence of a coalescent process is a limiting case of our coalescent-based model, our results allowing for mixtures of more general sequence evolution models extend those results in the ultrametric case. Algebraic study of a network model combined with the general Markov model, again with no coalescent process, was also conducted in \cite{CasanellasF-S2020}.

\medskip

This paper proceeds as follows. Section \ref{sec::networksmodels} defines the networks and models under consideration, as well as the logDet distance. Section 
\ref{sec::comb} uses combinatorial arguments to show how information on undirected rooted triple networks can be used to determine features of a larger directed network from which they are induced. Expected frequencies of site patterns for sequences produced by the mixture of coalescent mixtures model are studied in Section \ref{sec::freqs}, and shown to be expressible as convex combinations of pattern frequencies from simpler networks. In Section \ref{sec::logDet} we show that the ordering by magnitude of logDet distances for triples of taxa tells us about the induced rooted triple species network, and by combining this with the result of Section \ref{sec::comb} we obtain our main identifiability result, Theorem \ref{thm::main}. Section \ref{sec::other} further studies the logDet distances from a rooted triple network, in order to better understand what triples of distances can arise under the mixture of coalescent mixtures model. We conclude in Section \ref{sec:discussion} with an outline of how these results can be developed into a practical inference algorithm.

 \section{Networks and models}\label{sec::networksmodels}

 \subsection{Phylogenetic Networks}\label{sec::networks}
 Although there are many variations on the notion of a phylogenetic network in the literature, we adopt ones appropriate to the Network Multispecies Coalescent (NMSC) model. This model, which describes the formation of trees of gene lineages in the presence of both incomplete lineage sorting and hybridization, will be further developed in the next subsection. First, we focus on setting forth combinatorial aspects of the networks.

\begin{definition} \label{def::network} \cite{Solis-Lemus2016,Banos2019}
 	A  \emph{topological binary rooted phylogenetic network}  $\mathcal{N}^+$
 	on taxon set $X$ is a connected directed acyclic graph with  vertices   $V=V(\mathcal N^+)$ and edges $E=E(\mathcal N^+)$,  where $V$ is a disjoint union $V = \{r\} \sqcup V_L \sqcup V_H \sqcup V_T$  and $E$ is a disjoint union $E = E_H \sqcup E_T$, with a bijective leaf-labeling function $f : V_L \to X$ with the
 	following characteristics:
 	\begin{itemize}
 		\item[1.] The \emph{root} $r$ has indegree 0 and outdegree 2.
 		\item[2.] A \emph{leaf} $v \in V_L$ has indegree 1 and outdegree 0.
 		\item[3.] A  \emph{tree node} $v\in  V_T$ has indegree 1 and outdegree 2.
 		\item[4.] A \emph{hybrid node} $v\in  V_H$ has indegree 2 and outdegree 1.
 		\item[5.] A \emph{hybrid edge} $e=(v,w) \in E_H$ is an edge whose child node $w$ is hybrid.
 		\item[6.] A \emph{tree edge} $e=(v,w) \in E_T$ is an edge whose
 		child node $w$ is either a tree node or a leaf.
 	\end{itemize}
 \end{definition}

When $|X|=3$ or $4$, we refer to $\mathcal{N}^+$ as a \emph{rooted triple network} or a \emph{rooted quartet network}, respectively.

The vertices, and edges, of $\mathcal N^+$ are partially ordered by the directedness of the graph. For instance,  a node $u$ is \emph{below} a node $v$, and $v$ is \emph{above} $u$, if there exists a non-empty directed path in $\mathcal N^+$ from $v$ to $u$. The root is thus above all other nodes.

\medskip

A metric notion of the network above incorporates some of the parameters of the NMSC model. This introduces edge lengths, measured in generations throughout this article, as well as probabilities that a gene lineage at a hybrid node follows one or the other hybrid edge as it traces back in time toward the network root. Since we focus on binary networks, only hybrid edges are allowed to have length 0, to model possibly instantaneous jumping of a lineage from one population to another.

\begin{definition} A  \emph{metric binary rooted phylogenetic network}  $( \mathcal N^+, \{\ell_e\}_{e\in E},\{\gamma_e\}_{e\in E_H})$ is a topological binary rooted phylogenetic network together with an assignment of weights or \emph{lengths} $\ell_e$ to all edges and \emph{hybridization parameters} $\gamma_e$ to all hybrid edges, subject to the following restrictions:
\begin{itemize}
\item[1.] The length $\ell_e$ of a tree edge $e\in E_T$ is positive.
\item[2.] The length $\ell_e$ of a hybrid edge $e\in E_H$ is non-negative.
\item[3.] The hybridization parameters $\gamma_e$  and $\gamma_{e'}$ for a pair of hybrid edges $e,e'\in E_H$ with the same child hybrid node are positive and sum to 1.
\end{itemize}
\end{definition}

A metric network of this sort is said to be \emph{ultrametric} if every directed path from the root to a leaf has the same total length. This is equivalent to requiring the ultrametricity of
all trees displayed on the network. An example of a simple ultrametric network
is shown in Figure \ref{fig::network} (Right).

 \begin{figure}
	\begin{center}
	\includegraphics[width=3.8in]{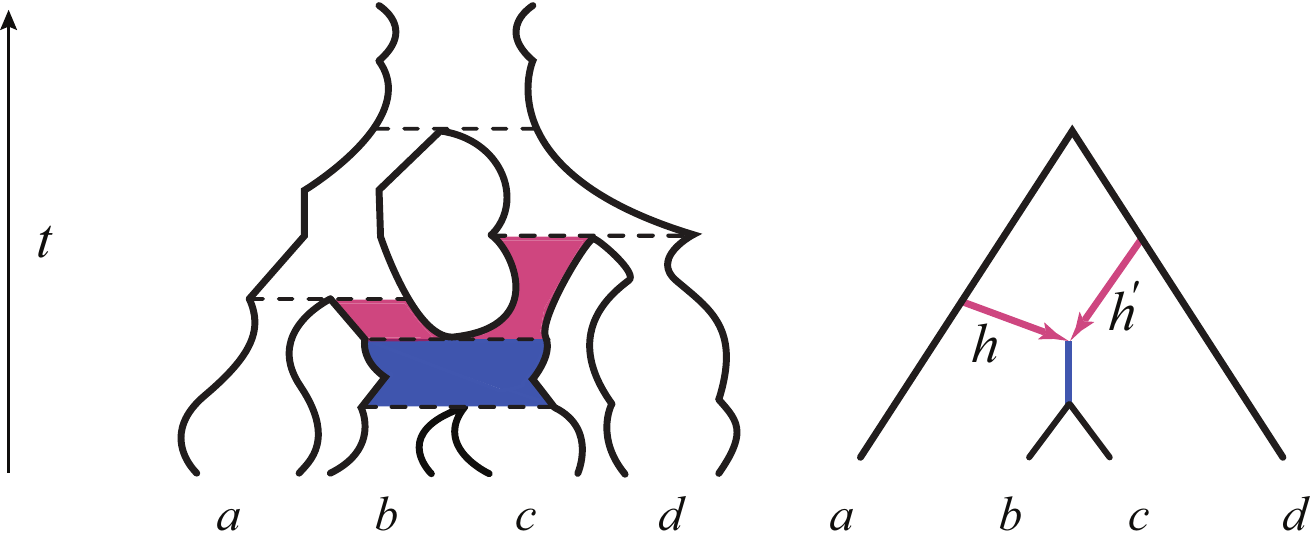}
	\end{center}
	\caption{ (Left) An ultrametric species network $\mathcal N^+$ with time $t$ in generations before the present, hybrid edges $h$ and $h'$ shown in red, and population functions $N_e(t)$ on each edge depicted by widths of ``tubes.''  The edge lengths $\tau$ are measured on the $t$-axis between the dashed lines indicating speciation and hybridization events. The dashed red/blue boundary represents a hybrid node, the top dashed line the root of the network, and other dashed lines tree nodes.  (Right) A schematic of the same species tree, which does not show population sizes. Hybridization parameters $\gamma$ and $\gamma'$ are omitted from both drawings.}\label{fig::network} 
\end{figure}

On directed networks there are several analogs \cite{Steel2016} of the most recent common ancestor of a set of taxa on a tree. The following is the most useful in this work.

\begin{definition}\label{def::LSA}\cite{Steel2016}
Let $\mathcal{N}^+$ be a (metric or topological)  binary rooted phylogenetic network on a set of taxa $X$  and let $Z\subseteq X$. Let $D$  be the set of 
nodes which lie on every directed path from the root $r$ of $\mathcal{N}^+$ to any $z\in Z$. Then the \emph{lowest stable
ancestor of $Z$ on $\mathcal{N}^+$}, denoted $\LSA(Z,{\mathcal{N}^+})$,  is the unique node $v\in D$  such that $v$ is below all $u\in D$ 
with $u\neq v$. The \emph{lowest stable ancestor (LSA)} of a network on $X$ is $\LSA(X)$.
\end{definition}

Phylogenetic networks as defined here have no cycles in the usual sense for a directed graph. The term \emph{cycle} will thus be used to refer to a collection of edges that form a cycle when all edges are undirected. A cycle must contain at least two hybrid edges sharing a hybrid node, and may contain any non-negative number of tree edges. The class of networks we focus on is those in which cycles are separated, in the following sense.

\begin{definition}\label{def::level1} 
	A rooted binary phylogenetic network $\mathcal{N}^+$ is said to be \emph{level-1} if no two distinct cycles in $\mathcal{N}^+$ share an edge.
\end{definition}

Although this is not the standard definition of level-1 \cite{Rossello2009}, in the setting of binary networks it is equivalent. 

Each cycle on a level-1 phylogenetic network contains exactly one hybrid node and two hybrid edges with that node as a child. Thus there is a one-to-one correspondence between cycles and the  hybrid nodes they contain.  A cycle composed of $n$ edges, 2 of which are hybrid, is called an \emph{$n$-cycle}. If the cycle's hybrid node has $k$ leaf descendants, it is an \emph{$n_k$-cycle}. 

\medskip

Passing from a large network to one on a subset of the taxa is similar to the process for trees.

\begin{definition}
\emph{Suppressing a node} with both in- and out-degree 1 in a directed phylogenetic network means replacing it and its two incident edges with a single edge from its parent to its child. For a metric network, the new edge is assigned a length equal to the sum of lengths of the two replaced. If the outedge was hybrid, the new edge is also hybrid and retains the hybridization parameter. 

Similarly, suppressing a node of degree 2 between two undirected edges means replacing it and its two incident edges with a single undirected edge.
\end{definition}

\begin{definition}\label{def::triplet} 
Let   $\mathcal{N}^+$  be a (metric or topological)  binary rooted phylogenetic network on $X$ and let $Y\subset X$. The \emph{induced  rooted network}   $\mathcal{N}^+_{Y}$ on $Y$ is the network obtained from   $\mathcal{N}^+$ by retaining nodes and edges in every path from the root $r$ on $\mathcal N^+$ to any $y\in Y$, and then suppressing all nodes with in- and out-degree 1.  We then say  $\mathcal{N}^+$ \emph{displays} $\mathcal{N}^+_{Y}$.
\end{definition}

We need the notion of a \emph{rooted undirected network}, in which all edges have been undirected but the root retained. Note that if a rooted network is a tree, knowledge of the root alone is enough to recover the direction of every edge, so this notion is not useful in that setting. If cycles are present, knowledge of the root determines only the direction of every cut edge (an edge whose deletion results in a graph with two connected components), and edges directly descended from cut edges. Knowing the root and all hybrid nodes in an undirected level-1 network  does, however, determine the full directed network. 

Several other notions of networks induced from a directed one are needed.

\begin{definition}\label{def::undirected}
Let   $\mathcal{N}^+$  be a (metric or topological)  binary rooted phylogenetic network on $X$. 
\begin{enumerate}
\item   \cite{Banos2019}
The \emph{LSA network} $\mathcal N^{\oplus }$
induced from $\mathcal N^+$  is the network on $X$ obtained by
deleting all edges and nodes above $\LSA(X, \mathcal{N}^+)$, and designating $\LSA(X, \mathcal{N}^+)$ as the root node.

\item
The \emph{undirected LSA network} $\mathcal N^\ominus$ is the rooted network obtained from the LSA network $\mathcal{N}^\oplus$ by undirecting all edges.

\item  \cite{Banos2019}
The \emph{unrooted semidirected network} $\mathcal N^-$ is the unrooted network obtained from the LSA network $\mathcal{N}^\oplus$ by undirecting all tree edges and suppressing the root, but retaining directions of hybrid edges.

\end{enumerate}
\end{definition}

For a binary level-1 network $\mathcal N^+$, the only possible structure above the LSA has the form of a (possibly empty) chain of 2-cycles \cite{Banos2019}, an example of which is shown in Figure \ref{fig::chain2cyc}. The LSA network $\mathcal N^\oplus$ is obtained by simply deleting that chain.

 \begin{figure}
	\begin{center}
	\includegraphics{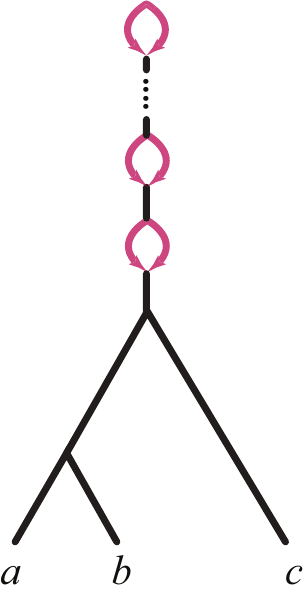}
	\end{center}
	\caption{ A rooted network $\mathcal N^+$ whose LSA network $\mathcal N^\oplus$ is the rooted tree $((a,b),c)$,  but which has a chain of 2-cycles above  $\LSA(a,b,c)$. }\label{fig::chain2cyc}
\end{figure}

Note that the terminology of ``$n_k$-cycles" can be applied to LSA networks $\mathcal N^\oplus$, as hybrid edges retain their direction.
On undirected LSA networks $\mathcal N^\ominus$, however,  ``$n$-cycle" can still be applied, but ``$n_k$-cycle" generally cannot.

\begin{definition}
By \emph{suppressing a cycle} $C$ in a topological level-1 network we  mean deleting all edges in $C$, identifying all nodes in $C$, and if the resulting node is of degree 2 suppressing it. If the network is rooted and this results in the root becoming a degree 1-node, then the resulting edge below the root is also deleted, with its child becoming the root. \end{definition}

Suppressing an $n$-cycle in a binary level-1 network results in a non-binary network when $n\ge 4$. However if only 2- and 3-cycles are suppressed, the result is binary.

\subsection{Coalescent Model on Networks}\label{ssec:NMSC} The 
formation of gene trees within a species network, as ancestral lineages of sampled loci from extant taxa join together moving backwards in time, is given a
mechanistic description by the Network Multispecies Coalescent Model (NMSC) \cite{Meng2009,Yu2011,Zhu2016}. 

Parameters of the NMSC for a set of taxa $X$
include a metric rooted binary phylogenetic network $(\mathcal N^+,\{\ell_e\},\{\gamma_e\})$ on $X$, with edge lengths $\ell_e$  in generations. In addition, for each edge $e=(u,v)$ fix a function $N_e:[0,\ell_e)\to \mathbb R^{>0}$ giving the (haploid) population size along the edge, where $N_e(0)$ is the population size at the child node $v$ and $N_e(t)$ is the population at time $t$ units above it. Finally, let $N_r:[0,\infty)\to \mathbb R^{>0}$ be an additional population size function for an infinite length `edge' ancestral to the root $r$ of the network.   
The $N_e$ need not be constant nor equal, although those are common assumptions in other works.  As in \cite{Allman2019} we make the biologically-plausible technical assumptions that the functions $N_e$ are bounded, and that all $1/N_e(t)$ are integrable over finite intervals.

Figure \ref{fig::network} (Left) depicts an example species network that is ultrametric in generations, with  hybrid edges $h$ and $h'$, and population functions $N_e$ on each edge depicted by time-varying widths of the network edges. The edge lengths $\ell_e$ are measured on the $t$-axis between the horizontal lines indicating speciation and hybridization events. Figure \ref{fig::network} (Right) gives a schematic of the same species tree, without a depiction of population functions. 

The standard Kingman coalescent models the formation of gene trees, with edge lengths in generations, within a single population edge $e$, with pairs of lineages coalescing independently as they trace backward in time, at instantaneous rate $1/N_e(t)$.  The multispecies coalescent model (MSC) extends this to a tree of populations, by using the standard coalescent on each edge, as well as an infinite length edge above the root, allowing multiple gene lineages to enter a population from its descendant ones at a tree node. The NMSC extends this further, so that
lineages reaching hybrid nodes randomly enter one or the other hybrid edge above them, with the choice determined independently according to the hybridization parameter probabilities. Thus the NMSC parameters $(\mathcal N^+,\{\ell_e\},\{\gamma_e\})$ and $\{N_e\}$ determine a distribution of rooted metric gene trees. The structure of the NMSC also ensures that  the distributions of gene trees obtained by marginalization to a subset $Y$ of taxa are the same as the distributions obtained from the NMSC on the displayed network $\mathcal N^+_Y$.

\subsection{Sequence substitution models on gene trees}\label{ssec::substitution}
The $k$-state \emph{general time-reversible model} (GTR) for sequence evolution is a continuous-time Markov process on a metric gene tree.  Gene tree edge lengths are in substitution units, and sequences are composed of $k$ possible states, or bases. Model parameters are a $k\times k$ instantaneous rate matrix  $Q$ together with a $k$-state distribution $ \pi $, with non-negative entries summing to 1, satisfying the following:
 \begin{enumerate}
 \item off-diagonals entries of $Q$ are positive,
 \item row sums of $Q$ are 0,
 \item $\operatorname{trace} Q=-1,$
 \item $\pi Q=0$,
 \item $\diag(\pi)Q$ is symmetric.
 \end{enumerate}
 In the ultrametric framework for our species networks, we introduce an additional time-dependent but lineage-independent rate scalar $\mu(t)$ for $Q$, where $t$ is measured in generations from leaves to the root and beyond, and $\mu(t)$ has units of substitutions/generation. We assume  $\mu$ is piecewise-continuous, $\mu(t)>0$ for all $t\ge 0$ so that the mutations process never stops, and $\int_0^\infty \mu(t)dt=\infty$ so that the total amount of possible mutation is unbounded. Following \cite{Allman2019}, this substitution model is denoted by GTR+$\mu$. 

For any node $u$ on a gene tree, let $t_u$ denote the distance, in generations, to that node from its descendant leaves.  The states at  a single site in sequences at the taxa at the leaves on the gene tree are then determined as follows: 
 A state is randomly chosen at the root of the tree from the distribution $ \pi$. For each edge $e=(u,v)$ descendant from a node $u$ the site undergoes random state changes with rates $\mu(t)Q$ for times $t\in[t_v,t_u]$ to obtain states at the child nodes. The full substitution process on the edge is thus described by the Markov matrix 
 $$M_e=\int_{t_v}^{t_u}\exp(\mu(t)Q)\,dt.$$
  A similar process is then repeated  for those nodes' children, and so on, until states at the taxa have been determined.
 
\subsection{Mixture of coalescent mixtures}\label{ssec::mixtures}
The model we focus on is the  $m$-class \textit{mixture of coalescent mixtures} \cite{Allman2019} on an ultrametric network. This model has as parameters  an ultrametric species network $( \mathcal N^+, 
\{\ell_e\},\{\gamma_e\})$, population size functions
$\{N_e\}$, a finite collection $\{(Q_i,{ \pi}_i; \mu_i)\}_{i=1}^m$ of GTR+$\mu$ parameters for the $m$ classes, and a vector $ \lambda$ of $m$ positive class size parameters summing to 1.  

Sequence data is generated as follows: For each site:
 \begin{enumerate}
      \item  a gene tree $T$ is sampled according to the NMSC model on $(\mathcal N^+, \{\ell_e\},\{\gamma_e\})$ with population sizes $\{N_e\}$, 
      \item  class $i$ is sampled  from the distribution $ \lambda$ to determine parameters $(Q_i,{ \pi}_i; \mu_i)$,	
      \item  the  bases for each $x\in X$ are sampled  under the  GTR+$\mu$ process on $T$ with parameters $(Q_i,{ \pi}_i; \mu_i)$.
\end{enumerate}

This model is denoted by $\M= \M(\theta)$ where $$\theta=( (\mathcal N^+, \{\ell_e\},\{\gamma_e\}),\{N_e\}, \lambda, \{(Q_i,{ \pi}_i; \mu_i)\}).$$ 
Sampling $n$ independent sites from this model produces $k$-state aligned sequences of length $n$. As usual in phylogenetics, these are summarized through counts of site patterns across the sequences in an $|X|$-dimensional $k\times k\times\cdots\times k$ array. Marginalizations of this array to 2-dimensions give pairwise $k\times k$ site pattern count matrices that compare only the sequences for two taxa in $X$.

\medskip

In the tree context, two extensions of this model were discussed in \cite{Allman2019}. For the first, the model assumption of one independently drawn gene tree for each site is modified to a more realistic one for genomic sequences in which all sites for a genetic locus share a gene tree. If the lengths (in number of sites) of the loci are independent identically distributed draws from some distribution, then the expected site pattern distribution for such a model is unchanged from that determined by $\mathcal M$. Only the rate of convergence, as the number of sampled genes grows, of frequencies of sampled site patterns to the asymptotic distribution will be slowed. Although 
the identifiability results of this paper are only formally stated for the model $\mathcal M$ on networks, they apply more generally to a similarly extended network model.

Another extension in the tree setting in \cite{Allman2019} allowed for relaxing the ultrametric condition while retaining strong results on identifiability from the logDet distances. In that extension, the scalar rate function was allowed to be edge dependent as long as a certain symmetry condition on mixture components resulted in ultrametricity in substitution units ``on average" across gene trees. While a similar model extension in the network setting seems likely to lead to similar results, it is not explored here, as the technical complications are greater than in the tree case.

\subsection{LogDet distance}\label{ssec:logDet}
The fundamental tool we use to study relationships of taxa under the mixture of coalescent mixtures model $\M$ is the logDet distance between a pair of aligned sequences. 
It is computed as follows: For taxa $a,b\in X$,  let $\widehat F^{ab}$ be 
a $k \times k$ matrix of
empirical relative site-pattern frequencies, obtained by normalizing the site pattern count matrix for $a$ and $b$, so that its entries sum to 1. Thus the $ij$ entry of $\widehat F^{ab}$ is the proportion of sites in the
sequences exhibiting base $i$ for $a$ and base $j$ for $b$. With $\hat f_a$ and $\hat f_b$ the vectors of row and column sums
of $\widehat F^{ab}$, which give the
proportions of various bases in the sequences for $a$ and $b$, let $\hat g_a$ and $\hat g_b$ the products of
the entries of $\hat f_a$, $\hat f_b$, respectively. Then  the empirical logDet distance is
\begin{align}\label{formula::logdet}
 \hat d_{LD}(a,b)=- \frac{1}{k}\left(\ln \det \left(\widehat F^{ab}\right )-\frac{1}{2}\ln (\hat g_a \hat g_b)\right) 
\end{align} 

Under most phylogenetic models, including the mixture of coalescent mixtures model, individual site patterns in sequences are assumed to be independent and identically distributed. By the weak law of large numbers, $\widehat F^{ab}$ computed from a sample will converge in probability to its expected value $F^{ab}$ as the sequence length goes to $\infty$. By the continuous function theorem (e.g., \cite{vanderVaart}), the empirical logDet distance thus converges in probability to the logDet distance computed by the same formula from the expected $F^{ab}$, a quantity we refer to as the \emph{theoretical logDet distance} and denote by $d_{LD}(a,b)$.

\section{Rooted Networks from Undirected Rooted Triple Networks}\label{sec::comb}

The goal of this section is to establish Proposition \ref{prop::combnet}, a combinatorial result indicating features of a topological level-1 rooted $n$-taxon network that can be recovered from its induced undirected rooted triple networks with 2- and 3-cycles suppressed. This is a rooted analog of a key result  of \cite{Banos2019} relating unrooted semidirected networks and their induced undirected quartet networks. Later sections of this paper focus on identifying these rooted triple networks under the model $\M$.

There are several possible routes to Proposition \ref{prop::combnet}. One approach would be to follow the argument of the quartet analog, with modifications throughout due to the rooted setting. Another would be to imitate the alternate proof of the quartet result given in \cite{ABR2019}, based on an extension of the intertaxon quartet distance of \cite{Rhodes2019}, but instead using the rooted triple distance also introduced in that work. 
The argument presented here is shorter than these approaches, as it leverages information about undirected rooted triple networks to obtain information about undirected quartet networks, and then applies the theory of \cite{Banos2019}. 

The following result, extracted from the  proof of Theorem 4 of \cite{Banos2019},  will be used. In it, and throughout this work, by a network \emph{modulo 2- and 3-cycles} we mean the network obtained by suppressing all 2- and 3-cycles. Similarly, \emph{modulo directions of edges in 4-cycles} means that all edges in 4-cycles are undirected. As a result, which of the edges in a 4-cycle are hybrid, and  therefore which node is hybrid, is not indicated.

\begin{lemma}[\cite{Banos2019}]\label{lem:banos}
Let $\mathcal N^+$ be a level-1 rooted binary topological  phylogenetic network on $X$. Let $Q$ be the set of undirected quartet networks obtained from those displayed on $\mathcal N^+$ by unrooting, suppressing all cycles of size 2 and 3, and undirecting all edges.  Then modulo 2- and 3-cycles and directions of edges in 4-cycles,  the semidirected unrooted network $\mathcal N^-$ is determined by $Q$.  
\end{lemma}

In order to apply this to rooted triples, we first recall some combinatorial  properties of rooted triple and quartet networks.

\begin{lemma}[\cite{Banos2019}]\label{lem:numberofcycles} 
	Let  $\mathcal{Q}^-$ be a  level-1 unrooted semidirected binary quartet network. Then $\mathcal{Q}^-$  has no $k$-cycles for $k\geq 5$, and at most one $4$-cycle. If
$\mathcal{Q}^-$ has a 4-cycle, then it has neither $3$- nor $2_2$-cycles. If there is no $4$-cycle, then there are at most two $3$-cycles, with at most one of these a $3_2$-cycle.
\end{lemma}

Lemma \ref{lem:numberofcycles} can be used to characterize possible cycles in a rooted triple network, by attaching an outgroup at the root.
More specifically, by \emph{attaching an outgroup $o$ to the root}  of an $n$-taxon network on taxa $X$  with $o\notin X$ we mean identifying the root $r$ of the network with the node $r$ on an edge $(r,o)$ and undirecting all tree edges.
This gives a 
$(n+1)$-taxon unrooted semidirected network. The rooted triple networks displayed on the original network are then in one-to-one correspondence with induced semidirected quartet networks containing $o$ on the new network. This construction yields the following.
 
\begin{corollary}\label{lem::numcyctriplet}
 	Let $\mathcal N^+$ be a  level-1 binary  rooted triple network. Then $\mathcal N^+$ has no $k$-cycles for $k\geq 5$, and at most one $4$-cycle in which case there are no $3$- or $2_2$-cycles. If there is no $4$-cycle, then there are at most two $3$-cycles, with at most one of these a $3_2$-cycle.
 \end{corollary}

Considering a rooted quartet network $\mathcal Q^+$, and the impact of passing to its associated unrooted semidirected quartet network $\mathcal Q^-$, Lemma \ref{lem:numberofcycles} also immediately yields the following.

\begin{corollary}\label{cor::numcycquart}
	Let $\mathcal Q^+$ be a level-1 rooted binary quartet network. Then $\mathcal Q^+$ has no $k$-cycles for $k\geq 6$, and has at most a one 5-cycle or 4-cycle, but not both.
\end{corollary}
 
 We now catalog the rooted quartet networks with 4- or 5-cycles, modulo smaller cycles.
 
\begin{lemma}\label{lem::triplets4c}
Let $\mathcal Q^+$ be a level-1 binary rooted quartet network with one $4$-cycle or one $5$-cycle. Then modulo 2- and 3- cycles and up to taxon relabelling, the LSA network $\mathcal Q^\oplus$ is one of those shown in Figure \ref{fig::quartet4c5c}. Thus $\mathcal Q^+$ displays either 1, 2, or 3  rooted triples with a 4-cycle.
\end{lemma}
\begin{proof}
	Let $\mathcal Q^+$ be a rooted level-1 network on $\{a,b,c,d\}$ with a cycle $C$ of size $4$ or $5$. By  Corollary \ref{cor::numcycquart}, $C$ is the only cycle of size greater than 3. Figure \ref{fig::quartet4c5c} shows the topologies, up to taxon relabeling, of all the rooted quartet networks with a $4$- or $5$-cycle and no  $2$- or $3$-cycles, as determined by enumerating all possible locations for adding hybrid edges to a rooted 4-taxon tree. The top row of Figure  \ref{fig::quartet4c5c} shows the quartet networks with exactly one displayed rooted triple, on $\{a,b,c\}$, having a 4-cycle. The middle row shows the networks with exactly two displayed rooted triples, on $\{a,b,c\}$ and $\{a,b,d\}$, having a 4-cycle. The bottom row shows those  with exactly three displayed rooted triples, on $\{a,b,c\}$, $\{a,b,d\}$, and $\{a,c,d\}$, having a 4-cycle. \qed\end{proof}
	
\begin{figure}
\begin{center}

\includegraphics{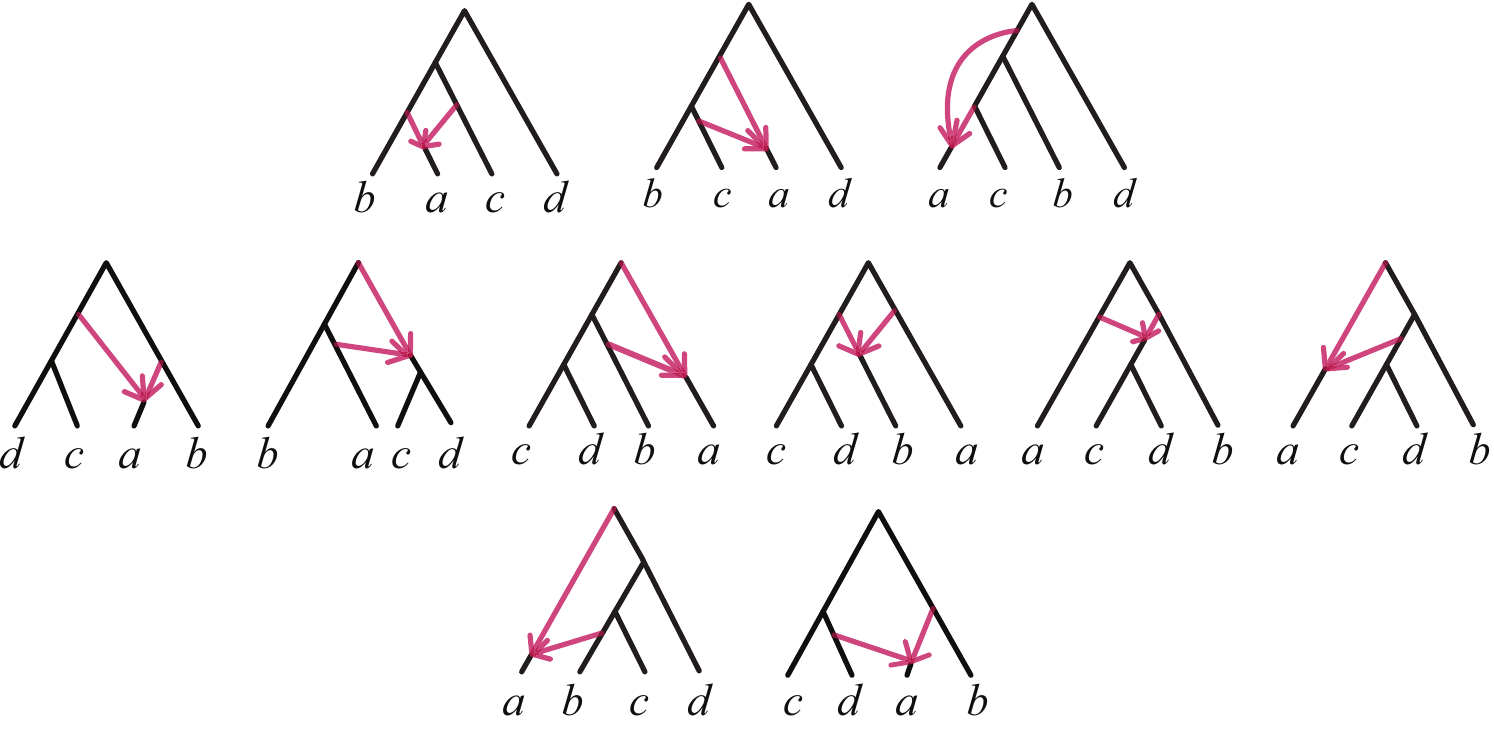}

\end{center}
\caption{All rooted directed topological quartet networks with a single $4$- or $5$-cycle, and no other cycles, up to relabeling of taxa. Networks in the top row display exactly one rooted triple with a 4-cycle, those in the middle row display two, and those in the bottom row display three.}\label{fig::quartet4c5c}
\end{figure}

Now we proceed to the main result of this section. 

\begin{proposition}\label{prop::combnet}
Let $\mathcal N^+$ be a level-1 rooted binary topological  phylogenetic network on $X$. Let $S$ be the set of undirected rooted triple networks obtained from those displayed on $\mathcal N^+$ by suppressing all cycles of size 2 and 3 and undirecting all edges.  Then modulo 2- and 3-cycles and directions of edges in 4-cycles, the LSA network $\mathcal N^\oplus$ is determined by $S$.  
\end{proposition}
\begin{proof}
We first build a set of rooted quartet networks from $S$. 
 Let $\{a,b,c,d\}\in X$  and let $S_{abcd}\subseteq S$ be the set of undirected rooted triple networks on any three elements of $\{a,b,c,d\}$, so $|S_{abcd}|=4$. 
By Corollary  \ref{cor::numcycquart} and Lemma \ref{lem::triplets4c}, there are $k=0$, $1$, $2$, or $3$ elements  of $S_{abcd}$ with a 4-cycle. We consider each possibility in turn, showing that we can determine the undirected rooted quartet network $N^\ominus_{abcd}$ modulo 2- and 3-cycles.  
 \smallskip
 
If $k=0$, all rooted triples in $S_{abcd}$ are trees and since $\mathcal N^+_{abcd}$ has no 4- or 5-cycles by Lemma \ref{lem::triplets4c}, the undirected LSA network $\mathcal N^\ominus_{abcd}$  modulo 2-and 3-cycles is a tree. By a well-known result for trees \cite{Semple2005}, $S_{abcd}$ determines $\mathcal N^\ominus_{abcd}$ modulo 2- and 3-cycles.

If $k=1$, then modulo 2- and 3-cycles and relabelling of taxa,
$\mathcal N^+_{abcd}$ is isomorphic to one of the networks in the top row of Figure \ref{fig::quartet4c5c}. But for all these networks if $a,b,c$ are the taxa in the rooted triple network with a 4-cycle, then the rooted 4-taxon network is obtained by attaching $d$ as an outgroup to it. Thus $\mathcal N^\ominus_{abcd}$ is determined modulo 2- and 3-cycles.

If $k=2$, $\mathcal N^+_{abcd}$ is isomorphic, modulo 2- and 3-cycles and relabeling, to one of the networks in the middle row of Figure  \ref{fig::quartet4c5c}. Note that for all those rooted quartet networks, the displayed rooted triple networks with 4-cycles are on $\{a,b,c\}$ and $\{a,b,d\}$, and the 4-taxon network can be obtained from either of these by replacing $c$ or $d$ with a cherry on $\{c,d\}$, thus determining $\mathcal N^\ominus_{abcd}$ modulo 2- and 3-cycles.
 
If $k=3$, $\mathcal N^+_{abcd}$ is isomorphic, modulo $2$-, and $3$-cycles and relabeling, to  one of the networks in the bottom row of Figure \ref{fig::quartet4c5c}.  In both of these, there is exactly one taxon, $a$,  that is in all three rooted triple networks with 4-cycles, and there is exactly one taxon, $c$,  that has graph-theoretic distance 3 from $a$ in exactly one of the two rooted triples with 4-cycles it appears in. 
Thus we can determine which taxon is $a$, and which is $c$. For the remaining pair $b,d$, if there is a taxon that is at distance 4 from  $a$ in both 4-cycle rooted triple networks it appears in, then the 4-taxon network is the one shown on the left, and that taxon is $d$. Otherwise, the network is the one shown on the right.
In this case there is exactly one rooted triple network on $a$ and $c$ which has its third taxon at distance 2 from the root, and this determines $b$.
Thus we obtain the rooted 4-taxon network $\mathcal N^\oplus_{abcd}$  modulo 2- and 3-cycles, and hence $\mathcal N^\ominus_{abcd}$  modulo 2- and 3-cycles
  
With all rooted 4-taxon networks  $\mathcal N^\ominus_{abcd}$ modulo 2- and 3-cycles determined, we attach an outgroup $o$ to all, giving the collection of all 5-taxon unrooted networks including $o$, modulo 2- and 3-cycles, induced from the unrooted network $\mathcal N'$ formed by attaching $o$ to the root of $\mathcal N^+$. But the unrooted 4-taxon networks displayed on these 5-taxon ones form the collection of all 4-taxon undirected networks (possibly including $o$) modulo 2- and 3-cycles displayed on $\mathcal N'$.

Lemma \ref{lem:banos} now determines $\mathcal N'$ modulo 2- and 3-cycles, with directions of cut edges and edges in cycles of size $\ge 5$, though not in 4-cycles.  Rooting $N'$ by the outgroup $o$ we recover the topology of $\mathcal N^\oplus$ modulo 2- and 3-cycles and directions of edges in 4-cycles.
\qed\end{proof}


\section{Expected pattern frequencies as convex sums}\label{sec::freqs}

The theoretical logDet distance between taxa depends on the matrix of expected relative site-pattern frequencies $F^{xy}$ in aligned sequences for taxa $x,y$, under the mixture of coalescent mixtures model $\mathcal M(\theta)$. The goal of this section is to show that $F^{xy}$ on a level-1 ultrametric rooted triple network can be expressed as a convex combination of frequency matrices for networks with no cycles below the LSA of the taxa. In this way, we reduce the computation of $F^{xy}$ to its computation on simpler networks. This is complicated somewhat by the fact that the convex combination may have terms which are expected pattern frequencies  conditioned on a pair of lineages coalescing below a certain node in a network. 

The lemmas that follow often involve modifying a network $\mathcal N^+$ by removing a hybrid edge, to obtain a new network $\mathcal N^+_i$. If one hybrid edge in a cycle is removed, the hybrid node is then suppressed as the other hybrid edge is joined to the descendant tree edge and given the induced length and population size.
We  retain all other edge lengths and population sizes, as well as hybrid parameters for unaffected cycles.  The parameters for the substitution process describing sequence evolution on gene trees are also retained. If $\theta$ denotes the full set of parameters associated to $\mathcal N^+$, then $\theta_i$ denotes the  full set of parameters associated to $\mathcal N^+_i$ in this way.
Notation such as  $F^{xy}(\theta)$ or $F^{xy}(\theta_i)$ denotes the dependence of $F^{xy}$ on the parameters $\theta$ or $\theta_i$, which include the network $\mathcal N^+$ or $\mathcal N_i^+$.

\begin{figure}\label{fig:234cycles}
	\begin{center}
	\includegraphics[width=4.in]{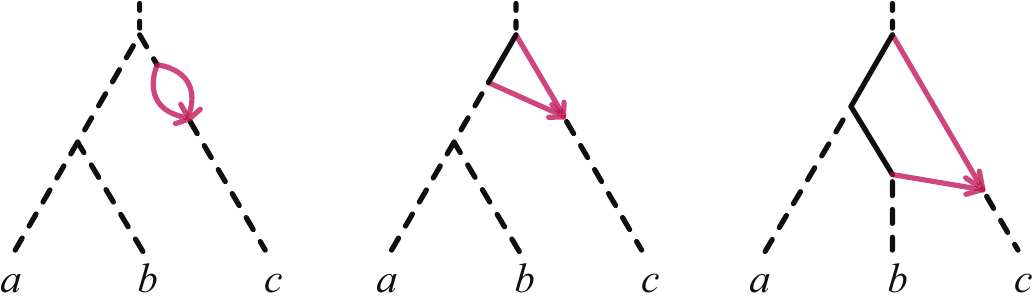}
	\end{center}
\caption{Examples of level-1 rooted triple networks with  $2_1$-, $3_1$-, and $4_1$-cycles. While multiple $2_1$-cycles may be present along any pendant edge shown here in dashes, there can be at most two $3_1$-cycles, whose hybrid nodes are located on a dashed pendant edge. At most one $4_1$-cycle can be present.
Site-pattern frequency matrices from the model $\mathcal M$ on rooted triple networks with these types of cycles are convex combinations of such matrices for 1, 2, or 4 networks without those cycles, as shown by Lemmas \ref{lem::2_1cycles} and \ref{lem::3_1and4cycles}.}
\end{figure}

The most straightforward network simplifications occur when the hybrid node of a cycle has a single descendant leaf, as depicted by the  example $2_1$-, $3_1$- and $4_1$-cycles in Figure \ref{fig:234cycles}. 

\begin{lemma}\label{lem::2_1cycles} (Removing $2_1$-cycles)
	Let $\mathcal N^+$ be a binary level-1 ultrametric rooted triple network on $\{a,b,c\}$ and let $C$ be a $2_1$-cycle in $\mathcal N^+$ with hybrid edges $h_1,h_2$. Let $\mathcal N^+_1$  be the network obtained from $\mathcal N^+$ by removing $h_2$.   Then, under the model $\M$ for any $x,y\in\{a,b,c\}$,
	$$F^{xy}(\theta)= F^{xy}(\theta_1).$$
\end{lemma}
\begin{proof}	Since  the hybrid node of $C$ has only one descendant,  the combined coalescent and substitution process on $\mathcal N^+$ can be expressed as a linear combination of those processes on $\mathcal N^+_1,\mathcal N^+_2$,  weighted by $\gamma_1=\gamma(h_1),\gamma_2=\gamma(h_2)$.   That is, for any $x,y\in\{a,b,c\}$,
	 	$$F^{xy}(\theta)= \gamma_1 F^{xy}(\theta_1)+\gamma_2 F^{xy}(\theta_2).$$
But $\mathcal N^+_1$ and $\mathcal N^+_2$ only differ by $h_1$ and $h_2$ which have the same length, though possibly different population sizes. However, since only one lineage can be present in the population for those edges, those population sizes have no impact in model $\mathcal M$, so $F^{xy}(\theta_2)=F^{xy}(\theta_1)$. Since $\gamma_1+\gamma_2=1$, the claim follows.
\qed\end{proof}

If a network $\mathcal N^+$ has multiple $2_1$-cycles, then applying Lemma \ref{lem::2_1cycles} repeatedly gives $ F^{xy}(\theta)=F^{xy}(\widetilde \theta)$ where $\widetilde{\mathcal N}^+$ is a rooted network with no $2_1$-cycles obtained from $\mathcal N^+$ by deleting one hybrid edge in each of the $2_1$-cycles on $\mathcal N^+$.

\begin{lemma}\label{lem::3_1and4cycles} (Decomposing $3_1$- and $4_1$-cycles)
	Let $\mathcal N^+$ be a binary level-1 ultrametric rooted triple network on $\{a,b,c\}$ and let $C$ be either a $3_1$- or a $4_1$-cycle on $\mathcal N^+$. Let $h_1,h_2$  be   the  hybrid edges of $C$ with $\gamma_i=\gamma(h_i)$. Let $\mathcal N^+_i$  be the network obtained from $\mathcal N^+$ by removing $h_j$, $j\ne i$.  Then, under the model $\M$ for any $x,y\in\{a,b,c\}$,
	$$F^{xy}(\theta)= \gamma_1 F^{xy}(\theta_1)+\gamma_2 F^{xy}(\theta_2).$$
\end{lemma}
 \begin{proof}
 	 	Since  the hybrid node of $C$ has only one descendant, we can express the combined coalescent and substitution process on $\mathcal N^+$ as a linear combination of the processes of the $\mathcal N_i$, with coefficients $\gamma_i$, $i=1,2$. 
		\qed\end{proof}
 
A level-1 rooted triple network may have one $4_1$-cycle, one $3_1$-cycle, or two $3_1$-cycles. In the last case, Lemma \ref{lem::3_1and4cycles} may be applied twice, to express the pattern frequency matrix under the model as a convex combination of four such matrices for networks with no $3_1$-cycles.
	
With Lemma \ref{lem::2_1cycles} this shows that computation of the  matrix of relative site-pattern frequencies of a level-1 ultrametric  rooted triple network $\mathcal N^+$ reduces to cases where  there are no $2_1$-, $3_1$-, or $4_1$-cycles. The effects of $2_2$- and  $3_2$-cycles  are more complicated, however, as a coalescent event may or may not occur below the hybrid nodes of such cycles. 

The following definition facilitates studying the impact of such cycles. In it a node $p$ may be either an existing node or a new node introduced along an edge of a network, with appropriate division of the original edge length and population function. Although strictly speaking this second case passes out of the class of binary networks, we allow this only to simplify reference to intermediate states of the coalescent process.
 
\begin{definition} Let $K_p(\theta)$  be the random variable giving the number of lineages at node $p\in V(\mathcal N ^+)$ under the NMSC. With $X_p\subseteq X$ denoting the set of taxa below $p$,  $K_p(\mathcal N^+)$ has sample space $\left\{1,2,\dots,|X_p|\right\}$.  
\end{definition}

When $\theta$ is clear from context we write  $K_p=K_p(\theta)$. We also use the notation  $F^{xy}_{|K_p=m}(\theta)$ to denote the joint distribution of site patterns conditioned on $K_p=m$ under the model $\M$ with parameters $\theta$.
 
 \begin{figure}
\begin{center}
 \includegraphics[width=4.in]{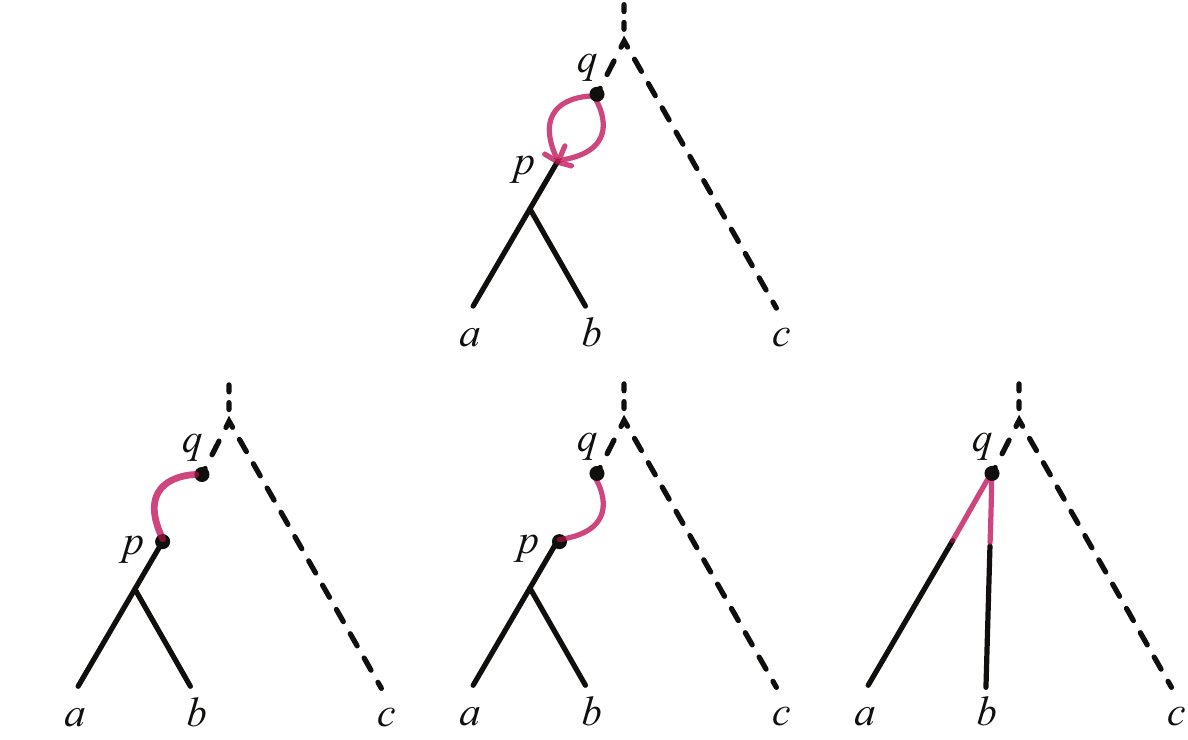}
\end{center}
 \caption{ (Top) A rooted level-1 ultrametric network on $\{a,b,c\}$, with the $2_2$-cycle closest to LSA($a,b$) shown. (Bottom) The networks $\mathcal N^+_1$, $\mathcal N^+_2$, and $\mathcal N^+_0$ obtained from $\mathcal N^+$, respectively, as described in Lemma
\ref{lem::2cyc}. Note that there may be additional cycles along the dashed lines, with hybrid nodes above node $q$ and taxon $c$.}\label{fig::2cyc}
 \end{figure}
 
\begin{lemma}\label{lem::2cyc}(Decomposing $2_2$-cycles)
 Let $\mathcal N^+$ be a binary level-1 ultrametric  rooted triple network on $\{a,b,c\}$ without $2_1$- or $3_1$-cycles. Suppose, as depicted in Figure \ref{fig::2cyc}, $C$ is a $2_2$-cycle on $\mathcal N^+$, with edges $h_1,h_2$ from node $q$ to hybrid node $p$, hybridization parameters $\gamma_i=\gamma(h_i)$, leaf descendants $a,b$ of $p$, and no cycles below $p$.   Denote by $\mathcal N^+_i$, $i=1,2$ the network obtained from $\mathcal N^+$ by removing $h_j$, $j\ne i$ and by $\mathcal N^+_0$ 
 the network obtained from $\mathcal N^+$ by deleting all edges and nodes below $q$ and attaching edges $(q,a)$ and $(q,b)$ of appropriate length so that $\mathcal N^+_0$ is ultrametric. Then, under the model $\M$ for any $x,y\in\{a,b,c\}$,
 \begin{align*}
F^{xy}(\theta)= & \gamma_1^2 F^{xy}(\theta_1 )+\gamma_2^2 F^{xy}(\theta_2 )+ P(K_p=2)2\gamma_1\gamma_2F^{xy}(\theta_0 )+P(K_p=1)2\gamma_1\gamma_2F^{xy}_{|K_p=1}(\theta_1).
\end{align*} 
\end{lemma}  
\begin{proof}
Since the structure of the model for $\mathcal N^+$, $\mathcal N^+_1$, and $\mathcal N^+_2$ is identical below $p$, we may also use $K_p$ to denote
$K_p(\theta_1)$ and  $K_p(\theta_2)$. Thus
 \begin{align}
  F^{xy}(\theta)&=P(K_p=2) F^{xy}_{|K_p=2}(\theta)+P(K_p=1) F^{xy}_{|K_p=1}(\theta)\notag\\
  		 &= P(K_p=2)\left [\gamma_1^2 F^{xy}_{|K_p=2}(\theta_1)+\gamma_2^2 F^{xy}_{|K_p=2}(\theta_2)+2\gamma_1\gamma_2F^{xy}(\theta_0)\right ] +	P(K_p=1) F^{xy}_{|K_p=1}(\theta ).\label{eq:2cyceq}
  \end{align}
 
But since $F^{xy}_{|K_p=1}(\theta)=F^{xy}_{|K_p=1}(\theta_i)$  for $i=1,2$ by the argument used for Lemma \ref{lem::2_1cycles}, and the identity $1=\gamma_1^2+\gamma_2^2+2\gamma_1\gamma_2$, 
$$
F^{xy}_{|K_p=1}(\theta)=\gamma_1^2 F^{xy}_{|K_p=1}(\theta_1)+ \gamma_2^2  F^{xy}_{|K_p=1}(\theta_2)+2\gamma_1\gamma_2F^{xy}_{|K_p=1}(\theta_1).
$$
Substituting this into equation \eqref{eq:2cyceq} and using $P(K_p=1)+P(K_p=2)=1$ yields the claim.
\qed\end{proof}
 
Note that while $\mathcal N^+_1$ and $\mathcal N^+_2$ of Lemma \ref{lem::2cyc} have the same topology and edge lengths, the hybrid edges $h_1,h_2$ may have different population sizes. Thus $F^{xy}(\theta_1 )\ne F^{xy}(\theta_2 )$ is possible. This is in contrast to the argument on removing $2_1$-cycles in Lemma \ref{lem::2_1cycles}, in which hybrid edge population sizes did not play a role.
 
Since a level-1 3-taxon rooted network cannot have a $2_2$-cycle above a $3_2$-cycle,  Lemma \ref{lem::2cyc} can be applied recursively to the $\mathcal N^+_i$, $i\in\{1,2\}$ to eliminate all $2_2$-cycles. Thus the remaining complication to producing an
expression
for   $F^{xy}(\theta)$ as a convex combination of such matrices for networks without $2_1$-, $3_1$-, or $2_2$-cycles is the presence of
terms of the form $F^{xy}_{|K_p=1}({\theta}')$ where ${\mathcal N'^+}$ has cherry $\{a,b\}$ and neither $2_1$- nor $3_1$-cycles. Such terms are handled with the following.

\begin{lemma} \label{lem::CondCoal} (Decomposing $2_2$-cycles conditioned on coalescence)
Let $\mathcal N^+$ be a binary level-1 ultrametric rooted triple network on $\{a,b,c\}$ without $2_1$-  or $3_1$-cycles, and 
on which $\{a,b\}$ form a cherry. Let $p$ be the  parent of the common parent of $a,b$. 
Denote by ${\widetilde {\mathcal N}^+}$  a network obtained from $\mathcal N^+$ by removing one hybrid edge from each $2_2$-cycle.

If $\mathcal N^+$ has no $3_2$-cycle, then
$$F^{xy}_{|K_p=1}(\theta)= F^{xy}_{|K_p=1}({\widetilde\theta}).$$

If $\mathcal N^+$ has a $3_2$-cycle, with hybrid edges $h_1,h_2$ and hybridization parameters $\gamma_i=\gamma(h_i)$, then let ${\widetilde {\mathcal N}}^+_i$ be the network obtained from ${\widetilde {\mathcal N}}^+$ by removing $h_j$, $j\ne i$. Then
$$F^{xy}_{|K_p=1}(\theta)= \gamma_1F^{xy}_{|K_p=1}(\widetilde \theta_1)+\gamma_2F^{xy}_{|K_p=1}(\widetilde \theta_2).$$
\end{lemma}

\begin{proof}
Conditioned on $K_p=1$, there is only one lineage in any population above $p$ and below the hybrid node of a $3_2$-cycle, if such a cycle is present, or the LSA otherwise.  
Thus, as in the proof of  Lemma \ref{lem::2_1cycles}, no $2_2$-cycle will have any effect on the joint distribution.
If there is no $3_2$-cycle on $\mathcal N^+$ this yields the claim. If there is a $3_2$-cycle,  since only one lineage reaches the hybrid node of  the $3_2$-cycle,  we obtain the claim as in the proof of Lemma \ref{lem::3_1and4cycles}.
\qed\end{proof}
 
 \begin{lemma}\label{lem::3_2cyc} (Decomposing $3_2$-cycles)
 	Let $\mathcal N^+$ be a binary level-1 ultrametric  rooted triple network on $\{a,b,c\}$ with no cycles below its LSA except a  $3_2$-cycle $C$. 
	Let  $p$ denote the hybrid node of $C$, and $h_1,h_2$ the hybrid edges with hybridization parameters $\gamma_i=\gamma(h_i)$ and  lengths $y,z$, as depicted at the top of Figure \ref{fig::3_2c}. Let $\mathcal N^+_1$, $\mathcal N^+_2,$ $\mathcal N^+_3,$ and $\mathcal N^+_4$  be the networks derived from $\mathcal N^+$ shown at the bottom of Figure \ref{fig::3_2c}.  Then, under the model $\M$, for any $x,y\in\{a,b,c\}$, with $K_p=K_p(\theta)$,
\begin{equation*} 
\begin{split}
F^{xy}(\theta)=&  \gamma_1^2 F^{xy}(\theta_1)  +  \gamma_2^2  F^{xy}(\theta_2)+P(K_p=2) \gamma_1\gamma_2  \left ( F^{xy}(\theta_3) + F^{xy}(\theta_4)\right ) \\
&\ \ \ \ +P(K_p=1) \gamma_1 \gamma_2 \left (F^{xy}_{|K_p=1}(\theta_1)+F^{xy}_{|K_p=1}(\theta_2) \right ).\\
\end{split}
\end{equation*}
\end{lemma}
 
\begin{proof}	
	Observe that
	\begin{equation} 
\begin{split}
F^{xy}(\theta)&= P(K_p=2)F^{xy}_{|K_p=2}(\theta)+P(K_p=1)F^{xy}_{|K_p=1}(\theta)\\
&= P(K_p=2)\left  [ \gamma_1^2 F^{xy}_{|K_p=2}(\theta_1)  +  \gamma_2^2  F^{xy}_{|K_p=2}(\theta_2)   + \gamma_1\gamma_2 F^{xy}(\theta_3)   +\gamma_1\gamma_2 F^{xy}(\theta_4)\right ] \\
&\ \ \ \ \ +P(K_p=1)F^{xy}_{| K_p=1}(\theta).
\end{split}\label{eq:3_2cyc}
\end{equation}	
Since  $F^{xy}_{|K_p=1}(\theta)=\gamma_1 F^{xy}_{|K_p=1}(\theta_1)+\gamma_2 F^{xy}_{|K_p=1}(\theta_2)$
 and $\gamma_1+\gamma_2=1$,
$$ F^{xy}_{|K_p=1}(\theta)=\gamma_1^2 F^{xy}_{|K_p=1}(\theta_1)+ \gamma_2 ^2  F^{xy}_{|K_p=1}(\theta_2) +\gamma_1 \gamma_2 \left(F^{xy}_{|K_p=1}(\theta_1)+F^{xy}_{|K_p=1}(\theta_2)\right).
$$

Using this  and $P(K_p=1)+P(K_p=2)=1$
in equation \eqref{eq:3_2cyc} yields the claim.
\qed\end{proof}
	
 \begin{figure}
	\begin{center}
	\includegraphics{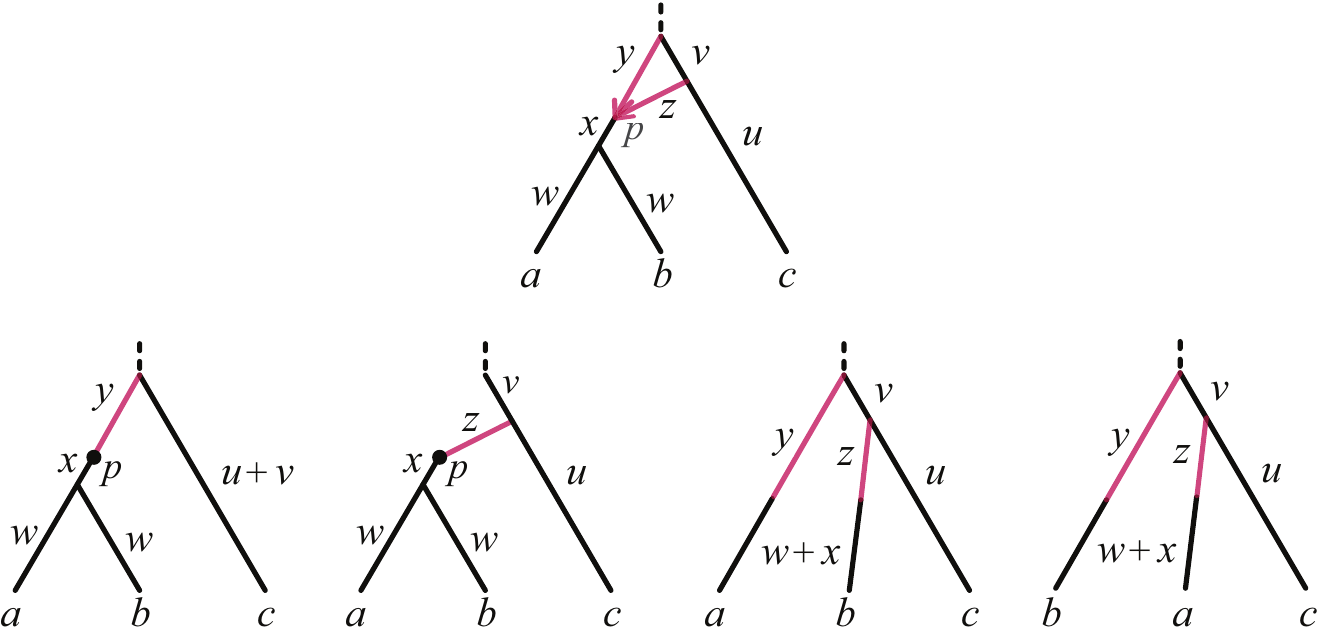}
	\end{center}
	\caption{ (Top)  A rooted level-1 ultrametric network with a $3_2$-cycle, and  (Bottom) the networks $\mathcal N^+_1$, $\mathcal N^+_2$, $\mathcal N^+_3$, and $\mathcal N^+_4$ used  in Lemma \ref{lem::3_2cyc}. Although only topology and branch lengths are shown, population size parameters for each edge of $\mathcal N^+_i$ are obtained from the corresponding ones of $\mathcal N^+$.
}\label{fig::3_2c}
\end{figure} 

\section{Theoretical logDet distances}\label{sec::logDet}

In this section, we show that, under the mixture of coalescent mixtures model $\mathcal M$ on an ultrametric level-1 rooted triple network, the theoretical logDet distances between taxa determine most topological features of the network.  The previous section established that the pattern frequency matrices for the model on such networks can be expressed as convex combinations of those on simpler networks (possibly subject to conditioning), whose only cycles are $2_3$-cycles located above  
$\LSA(a,b,c)$, such as depicted in Figure \ref{fig::chain2cyc}.  The following algebraic lemma is key to drawing conclusions about the determinants
of such linear combinations of matrices. 

\begin{lemma}[\cite{Allman2019}, Lemma 3.1]\label{lem::main}
	Suppose for each $i$, $F_i$ and $G_i$ are $\kappa\times \kappa$ symmetric positive definite matrices such that $y^TF_iy\geq y^TG_iy$ for every $y\in \RR^\kappa$ with the inequality strict for some $y$ and some $i$. For $\alpha_i\geq 0$, let
	$$ F=\sum_{i=1}^{m}\alpha_i F_i,\quad G=\sum_{i=1}^{m}\alpha_i G_i.$$
	Then 
	$$ \det F> \det G.$$
\end{lemma}

Analyzing the pattern frequency matrix for networks with $2_3$-cycles above $\LSA(a,b,c)$ requires a detailed look at the coalescent process
in such a chain of $2$-cycles.
For a simple case, assume lineages $x$ and $y$ enter the single cycle chain depicted in Figure \ref{fig::single2cycle}. 
Population functions  $N_1,$ $N_2$ , $N_3$, and $N_4$ are fixed for each edge, where for convenience, we shift domains from the convention in Section \ref{ssec:NMSC} so that $N_1$ is defined on $[0,t_0)$,  $N_2,N_3$ on $[t_0,t_1)$, and $N_4$ on $[t_1,\infty)$.

The probability density $c(t)$ for time to coalescence of the lineages $x,y$ entering at the bottom node ($t=0$) can be calculated piecewise as follows:
For $t\in[0,t_0)$, 
$$c(t)=\frac 1{{ N}_1(t)} \exp \left ( -\int_0^t \frac 1{{ N}_1(\tau)}\, d\tau\right ),$$
as given in \cite{Allman2019}.

For $t\in[t_0,t_1)$, 
$$c(t)={p_0}\left( \gamma^2 c_2(t) +(1-\gamma)^2 c_3(t)\right )$$
where $p_0=1-\int_0^{t_0} c(t)\,dt$ is the probability of no coalescence before $t_0$, and for $i=2,3$
$$c_i(t)=\frac 1{N_i(t)} \exp \left (-\int_{t_0} ^t \frac 1{N_i(\tau)}\, d\tau\right).$$

Finally, for $t\in[t_1,\infty)$, with $p_1=1-\int_0^{t_1} c(t)\,dt$ the probability of no coalescence before $t_1$,
$$c(t)={p_1} \frac 1{N_4(t)} \exp \left (-\int_{t_1}^t \frac 1{N_4(\tau)}\, d\tau\right).$$

 \begin{figure}
	\begin{center}
	\includegraphics[width=5cm]{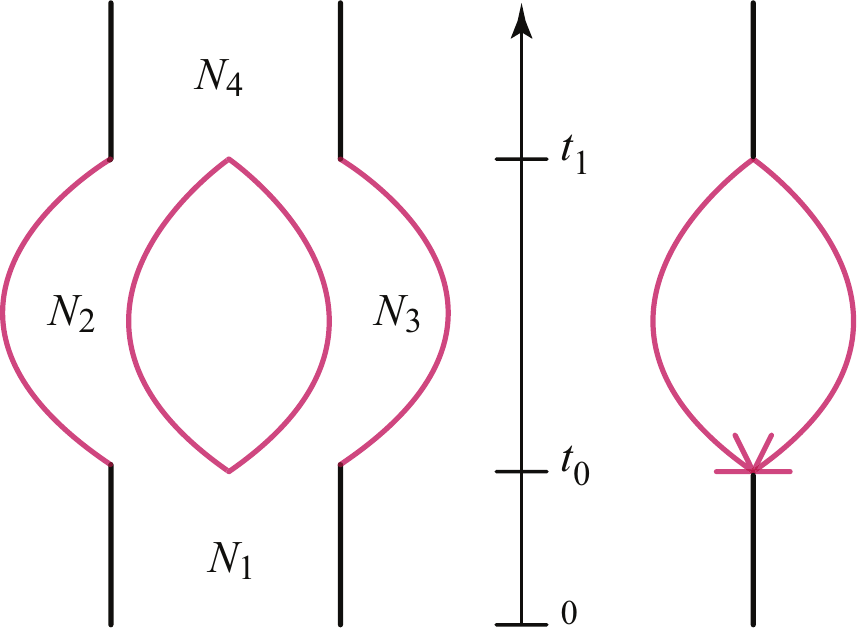}
	\end{center}
	\caption{A 2-cycle and adjacent tree edges in a species network, depicted (Left) with pipes whose width represent population sizes, and  (Right) as a schematic. }\label{fig::single2cycle}
\end{figure}

It is straightforward to extend this analysis of $c(t)$ to a chain with an arbitrary number of 2-cycles. 
Since we will not need an explicit formula for the distribution of coalescent times for two lineages entering such a chain of 2-cycles, we omit a complete derivation, and only state the properties of it that we use.

Formally, a \emph{chain of 2-cycles} is a species network with leaf  $a_0$, internal vertices  $b_1$, $a_1$, $b_2$, $a_2,\dots, a_n$, with root $r=a_n$, tree edges $e_i=(b_i,a_{i-1})$,   and hybrid edges 
$e_i'=(a_i,b_i)$, $e_i''=(a_i,b_i)$, together with edge lengths, piecewise-continuous population size functions on each edge, including above the root, and hybrid parameters $\gamma_i',\gamma_i''=1-\gamma_i'$ for each pair of hybrid edges $e_i',e_i''$. 
 
Using the technical assumptions given in Subsection \ref{ssec:NMSC}, it is straightforward to deduce the following.

\begin{lemma}\label{lem::coalFacts} Consider a fixed chain of 2-cycles with leaf $a_0$.  Let $c:[0,\infty)\to \mathbb R^{\ge 0}$ denote the probability density function under the NMSC for the time $T$ of coalescence of two lineages entering the chain at $a_0$. Then $c(t)$ is piecewise 
continuous, and $c(t)>0$ for all $t\in[0,\infty)$.
\end{lemma}
 
The next three technical lemmas generalize Lemmas 4.1, 4.4, and 4.5 of \cite{Allman2019} from a tree to a network setting. These culminate in Proposition \ref{prop::chain2cyc} below, which justifies the application of Lemma \ref{lem::main}. 

\begin{lemma} \label{lem:Cineq} Let $c:[0,\infty)\to \mathbb R^{\ge 0}$ be the probability density function under the NMSC for the time $T$ of coalescence of two lineages entering a chain of 2-cycles, and for times $t_2>t_1\ge 0$
let $c_i$ be the conditional density  given $T\ge t_i$.  Then the cumulative distribution functions for $c_1$ and $c_2$
satisfy
$$C_1(t)\ge C_2(t),$$
with the inequality strict on some interval.
\end{lemma}
\begin{proof}
Since $0=c_2(t)\le c_1(t)$ for all $t\le t_2$, the inequality is immediate for $t\le t_2$. Since using Lemma \ref{lem::coalFacts} we have $c_1(t)>c_2(t)=0$ for $t\in(t_1,t_2)$,
the inequality is strict on a subinterval.

For $t\ge t_2$, let $J=\int_{t_1}^{t_2} c_1(t) \,dt$ and $I(t)=\int_{t_2}^t c_1(s) \,ds$, so 
\begin{align*}
C_1(t)-C_2(t)&=J+I(t)-\frac {I(t)}{1-J}\\
&= J-\frac J{1-J} I(t).
\end{align*}
Differentiating and using Lemma \ref{lem::coalFacts} shows $C_1(t)-C_2(t)$ is decreasing for $t>t_2$. Since $C_1(t)-C_2(t)\to 0$ as $t\to\infty$, this implies $C_1(t)-C_2(t)\ge 0$, as claimed.
\qed\end{proof}

\begin{lemma}\label{lem:inteig} Let $c_1,c_2$ be probability density functions on $[0,\infty)$, with cumulative distribution functions $C_1,C_2$, such that $C_1(t)\ge C_2(t)$ for all $t$, with the inequality strict on some interval. Let
$s(t)=\int_{0}^{t} \mu(x)\,dx$ for a positive, piecewise-continuous $\mu$ on $[0,\infty)$ such that $s(\infty)=\infty$.
For $\lambda\le 0$ let
$$f(\lambda,\mu, C_i)=\int_0^\infty \exp(2 \lambda s(t))c_i(t) \,dt.$$
Then if $\lambda=0$,
$$f(0,\mu, C_1)=f(0,\mu, C_2)=1.$$
while for $\lambda<0$
$$f(\lambda,\mu, C_1)>f(\lambda,\mu, C_2).$$
\end{lemma}
\begin{proof}
For $\lambda=0$ we find $f(0,\mu, C_i)=\int_0^\infty c_i(t)\,dt=1$.

If $\lambda<0$, integrating by parts yields
\begin{align*}
f(\lambda,\mu, C_i) &= \exp(2 \lambda s(t))C_i(t)\bigg |_{t=0}^\infty -  2\lambda   \int_0^\infty \mu(t) \exp(2 \lambda s(t))C_i(t) dt\\
&=  -  2\lambda   \int_0^\infty \mu(t) \exp(2 \lambda s(t))C_i(t) dt.
\end{align*}
Thus
$$
f(\lambda,\mu, C_1)- f(\lambda,\mu, C_2)=  
 -2\lambda   \int_0^\infty  \mu(t) \exp(2 \lambda s(t)) (C_1(t)-C_2(t)) dt.
$$
As the integrand is non-negative, and positive on some interval, the claim for $\lambda<0$ follows. 
\qed\end{proof}

\begin{lemma}\label{lem:41}
Consider a GTR substitution model with rate matrix $Q\ne 0$, a scalar-valued rate function $\mu(t)$ satisfying the assumptions of Subsection \ref{ssec::substitution}, and a cumulative distribution function $C(t)$ for the time $T$ to coalescence of 2 lineages in a population. 

Let $F(x)=F(Q,\mu,C,x)$ be the expected site-pattern frequency array for two lineages that enter a population at time 0 and undergo substitutions at rate $\mu(t)Q$ conditioned on $T\ge x$. 
For $x<x_1$ let $\widetilde F(x,x_1) =\widetilde F(Q,\mu,C,x,x_1)$ be the expected site-pattern frequency array for two lineages that enter a population at time 0 and undergo substitutions at rate $\mu(t)Q$ conditioned on $x< T< x_1$.

Then for all $0\ne y\in \mathbb R^k$ the functions $y^TF(x)y$  and $y^T\widetilde F(x,x_1)y$ are positive-valued and decreasing in $x$.
Moreover there exists a $y$ for which both are strictly decreasing, and for which if $x_0<x_1\le x_2$
$$y^T\widetilde F(x_0,x_1)y>y^TF(x_2)y.$$

\end{lemma}
\begin{proof} Let $c_x(t)$ denote the conditional probability density function for the coalescent time $T$ given $T>x$.
With $s(t)=\int_{0}^{t} \mu(\tau)\,d\tau$, the Markov matrix describing the substitution process on a single lineage from time $0$ to time $t$ is 
$$M(\mu,Q,t)=\exp\left (s(t)Q \right).$$
Thus using time-reversibility of the substitution process, with $\pi$ the stationary distribution for $Q$,
$$F(x)=\diag(\pi) \int_0^\infty  (M(\mu,Q,t))^2 c_x(t)\,  dt.$$
Here  the square of the Markov matrix accounts for substitutions in the two lineages before coalescence.

Now $S^{-1}QS$ is diagonal for a matrix $S=\diag(\pi)^{-1/2} U$ with $U$ orthogonal, and $Q$'s eigenvalues satisfy $0=\lambda_1\ge \lambda_2\ge\cdots\ge\lambda_k$ with at least one $\lambda_i<0$ (Lemma 2.2 of \cite{Allman2019}). Thus diagonalizing the Markov matrix yields
\begin{align*}
U^T\diag(\pi)^{-1/2} F(x)\diag(\pi)^{-1/2} U&=\int_0^\infty \Lambda_{M(\mu,Q,t)} c_x(t) \, dt
\end{align*}
where $\Lambda_{M(\mu,Q,t)}$ is diagonal with entries $\exp(2s(t)\lambda_i)$. The diagonal entries of this integral are thus
$$\int_0^\infty \exp(2s(t)\lambda_i) c_x(t) \, dt.$$
But Lemmas \ref{lem:Cineq} and \ref{lem:inteig}
show this is positive, decreasing in $x$, and strictly decreasing for some $i$.
This establishes the claims about $F$, by choosing $y$ to be any eigenvector of $Q$ whose eigenvalue is negative to obtain a strictly decreasing function.

The corresponding claims about $\widetilde F$ are given by the same argument with the cumulative distribution function $C$ replaced by the conditional distribution function given the coalescent time $T<x_1$, that is, with
$$\widetilde C_{x_1} (t) =\begin {cases}
C(t)/C(x_1) & \text{if $t\le x_1$}\\
1&\text{ if $t>x_1$}
\end{cases}.$$

Finally, since for every $t$ the function $\widetilde C_{x_1} (t)$ is decreasing in $x_1$, then for any $y$ and $x_0$, a similar diagonalization argument and again using Lemma \ref{lem:inteig} shows the function $y^T\widetilde F(x_0,x_1)y$ is decreasing in $x_1$. Thus
if $x_0<x_1\le x_2$, then
$$y^T\widetilde F(x_0,x_1)y\ge\lim_{x_1\to\infty} y^T\widetilde F(x_0,x_1)y=y^T F(x_0)y\ge y^T F(x_2)y.$$ 
Moreover, if $y$ is an eigenvector of $Q$ whose eigenvalue is negative, then strict inequality holds.
\qed\end{proof}

\begin{proposition}\label{prop::chain2cyc}
Let $\mathcal N^+$ be a binary level-1 ultrametric  rooted triple network on $\{a,b,c\}$ whose LSA network has topology $((a,b),c)$,  but above $\LSA(\{a,b,c\}, \mathcal N^+)$ there is possibly a chain of 2-cycles.Then, under a coalescent mixture model on $\mathcal N^+$ with fixed parameters $\mu(t)$, $\{N_e\}$, $Q$, $\pi$, the relative site-pattern frequency matrices $F^{ab}$, $F^{bc}$, and $F^{ac}$ are  symmetric positive definite, with $F^{ac}=F^{bc},$   and  satisfy
		$$y^TF^{ab}y\geq y^TF^{ac}y$$  for every $y\in \RR^k$, with the inequality strict for some $y$. 
Moreover, the same statements hold when the arrays $F^{xy}$ are replaced by $F^{xy}_{|K_p=1}$ with $p$ a node placed above the parent of $a,b$ and below the parent of $c$.
\end{proposition} 
 \begin{proof}
           Let $x_1$ be the length of the pendant edges to $a$ and $b$, and $x_2$ the length of the pendant edge to $c$, so $x_2>x_1$. Then applying  Lemma \ref{lem:41} for an appropriately chosen distribution $C(t)$ of coalescent times so
           $$F^{ab}=F(x_1),\ \ F^{ac}=F^{bc}=F(x_2),$$
the result is immediate.	
           
Let $x_p$ denote the distance from $a$ or $b$ to $p$,  so $x_1<x_p< x_2$. Then conditioning on $K_p=1$, in the notation of Lemma \ref{lem:41} we have
       $$F^{ab}_{|K_p=1}=\widetilde F(x_1,x_p),\ \ F^{ac}_{|K_p=1}=F^{bc}_{|K_p=1}=F^{bc}=F(x_2),$$ so again Lemma \ref{lem:41}  yields the claim.
 \qed\end{proof}

We now turn from considering a coalescent mixture model, with a single substitution model class, to the mixture of coalescent mixtures $\mathcal M$.

\begin{lemma}\label{lem::switchability}
	Let $\mathcal N^+$ be a level-1 ultrametric rooted triple network on $\{a,b,c\}$ with no $4$-cycle.  Suppose $\{a,b\}$ form a cherry in the tree topology obtained from suppressing all cycles of  $ \mathcal N^+$. Then, under the mixture of coalescent mixtures model $\M$ on $\mathcal N^+$,
	$ F^{ac}(\theta)= F^{bc}(\theta).$   	
\end{lemma}
 \begin{proof}
 	By Lemmas  \ref{lem::2_1cycles} and  \ref{lem::3_1and4cycles}, we may assume $\mathcal N^+$ has neither a $2_1$- nor a $3_1$-cycle, so there are no cycles below the parent of $a,b$. By the ultrametricity of the network,  $a$ and $b$ are exchangeable under the combined coalescent and substitution model for each substitution model class, and therefore for the model $\M$.
\qed\end{proof}

This result is used to show that logDet distances from rooted triple networks with only $2$- and $3_1$-cycles satisfy the same equality and inequality relationships as those from trees.
  	
\begin{proposition} (No $4_1$-cycles or $3_2$-cycles) \label{prop::2and3_1cyc}Let $ \mathcal N^+$ be a level-1 ultrametric rooted triple network on $\{a,b,c\}$ with neither a 4-cycle nor a $3_2$-cycle. Let $\T=((a,b),c)$ be the tree topology obtained after suppressing all cycles in  $ \mathcal N^+$.  Under the mixture of coalescent mixtures model $\M$ on  $ \mathcal N^+$ the theoretical logDet distances  satisfy 
$$d_{LD}(a,c)=d_{LD}(b,c)> d_{LD}(a,b).$$ 
\end{proposition}
\begin{proof} Under the model $\M$, the frequencies of  bases at any taxon are identical, given by the same convex combination of the base frequency vectors $ \pi_i$ for substitution classes $i$.  Thus the value of  $\ln(g_ug_v)$ in the definition of the logDet distance, equation \eqref{formula::logdet}, is identical for every pair of distinct taxa  $x,y\in\{a,b,c\}$.  
It thus suffices to show 
$$\det F^{ab}(\theta) \geq \det  F^{ac}(\theta) =\det F^{bc}(\theta).$$

Lemma \ref{lem::switchability} gives the equality.   By Lemmas \ref{lem::2_1cycles}, \ref{lem::3_1and4cycles}, and \ref{lem::2cyc}, we can express  $ F^{xy}(\theta)$ as a convex combination of relative site-pattern frequency matrices, possibly conditioned on $K_p=1$, of networks of the form of the tree $\T$ joined to a (possibly empty) chain of 2-cycles above $\T$'s  root, such as depicted in Figure \ref{fig::chain2cyc}. By  Proposition \ref{prop::chain2cyc} each of those matrices for coalescent mixture models satisfy the hypotheses of Lemma \ref{lem::main}. Lemma \ref{lem::main} thus yields the claim for mixtures of coalescent mixtures by considering a convex combination across both the networks and substitution model classes.\qed\end{proof} 

A weaker result, without the inequality, applies to networks with $3_2$-cycles.

\begin{proposition} ($3_2$-cycle) \label{prop::3_2cyc}  Let $ \mathcal N^+$ be a level-1 ultrametric rooted triple network on $\{a,b,c\}$ with a $3_2$-cycle. Let $\T=((a,b),c)$ be the tree topology obtained after suppressing all cycles in  $ \mathcal N^+$.  Then under the mixture of coalescent mixtures model $\M$ on  $ \mathcal N^+$,  the theoretical logDet distances  satisfy 
$$d_{LD}(a,c)=d_{LD}(b,c).$$  
\end{proposition}
\begin{proof}
From Lemma \ref{lem::switchability}, $ F^{ac}(\theta) =F^{bc}(\theta)$, so the result follows as in the previous proof.  
\qed\end{proof}

Proposition \ref{prop::2and3_1cyc},  and the arguments leading to it, show that the equality and inequality relationships of logDet distances between only 3 taxa carry no signal of either $2$- or $3_1$-cycles. Proposition \ref{prop::3_2cyc}, however, leaves open the possibility that for a network with a $3_2$-cycle the smallest distance may not necessarily correspond to the taxa which are neighbors after 2- and 3- cycles are suppressed. This suggests that the presence of a $3_2$-cycle might be detectable, at least under some circumstances. In Section \ref{sec::other} we return to this issue, providing a more in-depth analysis of triples of logDet distances.

\begin{proposition}\label{prop::4cyc} ($4_1$-cycle) Let $ \mathcal N^+$ be a level-1 ultrametric rooted triple network on $\{a,b,c\}$ with a $4$-cycle, such that contracting all cycles except the 4-cycle and then deleting one of its hybrid edges  gives the trees $((a,b),c)$ and $((a,c),b)$. (See Figure \ref{fig::4c}.)
 Then under the mixture of coalescent mixtures model $\M$ on  $ \mathcal N^+$,  the theoretical logDet distances  satisfy 
$$d_{LD}(b,c) > d_{LD}(a,b)\text{ and } d_{LD}(b,c) > d_{LD}(a,c).$$
Moreover, if all other parameters are fixed, then for generic values of the hybridization parameters,
$$ d_{LD}(a,b)\neq  d_{LD}(a,c).$$
\end{proposition}
\begin{proof} As in Proposition \ref{prop::2and3_1cyc}, to establish these inequalities for the logDet distance, it is enough to show
\begin{equation}\det F^{bc} (\theta)< \det F^{ab}(\theta)\text{ and }\det F^{bc}(\theta)< \det F^{ac}(\theta).\label{eq:detineq}
\end{equation}

From Lemmas \ref{lem::2_1cycles} and \ref{lem::3_1and4cycles}, for $x,y\in\{a,b,c\}$
	$$F^{xy}(\theta)= \gamma_1 F^{xy}(\theta_1)+\gamma_2 F^{xy}(\theta_2)$$
where $\mathcal N^+_1$ and $\mathcal N^+_2$ have the structure of the trees $((a,b),c)$ and $((a,c),b)$ with chains of 2-cycles possibly attached above their roots.
Proposition \ref{prop::chain2cyc} implies that for each GTR substitution model class
$$y^TF^{ab}(\theta_1)y\geq y^TF^{bc}(\theta_1)y=y^TF^{ac}(\theta_1)y\quad\text{ and }\quad y^TF^{ac}(\theta_2)y\geq y^TF^{ab}(\theta_2)y= y^TF^{bc}(\theta_2)y,$$
 for every $y\in \RR^k$, with the inequalities strict for some choices of $y$. 
From this and Lemma \ref{lem::main} we obtain the inequalities \eqref{eq:detineq}.

To see $ d_{LD}(a,b)\neq  d_{LD}(a,c)$ for generic hybridization parameters, first observe that these distances extend to analytic functions of the $\gamma$ on all of $\mathbb C$. To show the inequality for generic $\gamma$, it is enough to show there exists one specific choice of $\gamma \in \mathbb C$ for which they are not equal. First consider a choice on the boundary of the  parameter space, by letting $\gamma_e=1$, $\gamma_{e'}=0$ for every pair $e,e'$ of hybrid edges with a common child so that the model 
reduces to one on the tree $((a,c),b)$.
In this case Theorem 1 of
 \cite{Allman2019} establishes the inequality. Continuity implies that there are then choices of $0<\gamma_e<1$, where the model does not degenerate to one on a tree, for which these distances are also not equal.
\qed\end{proof}

Assuming generic parameter values, Proposition \ref{prop::4cyc} combined with earlier results implies that the presence of a $4$-cycle is indicated by three distinct logDet distances computed from expected pattern frequencies. However, the three networks at the top of Figure  \ref{fig::4c} all satisfy the hypothesis of Proposition \ref{prop::4cyc}, but using equalities and inequalities of
logDet distances we cannot distinguish them. We can only identify their undirected version as depicted in the bottom of Figure \ref{fig::4c}.  

\begin{figure}
	\begin{center}
	\includegraphics[width=4.in]{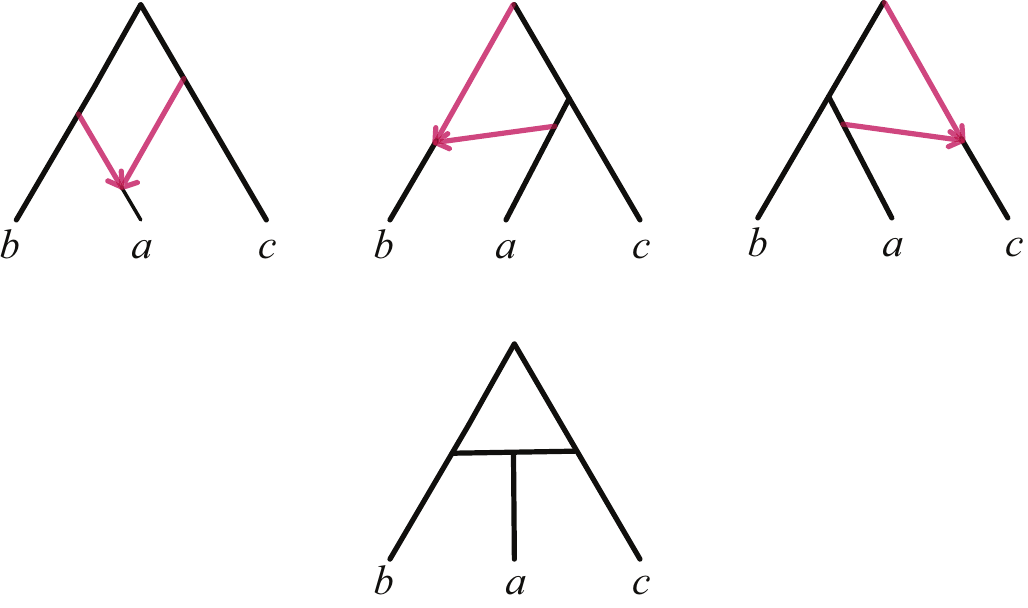}
	\end{center}
	\caption{ (Top) Three topologically-distinct rooted triple networks with a 4-cycle  displaying the trees $((a,b),c)$ and $((a,c),b)$. (Bottom) The undirected rooted topology shared by them. }\label{fig::4c}
\end{figure}

Nonetheless, the combinatorial result of Proposition \ref{prop::combnet} yields information on larger cycles and their hybrid nodes by first using logDet distances to determine undirected rooted triple networks. This gives our main result.

\begin{theorem}\label{thm::main}Let $ \mathcal N^+$ be a binary level-1 ultrametric  network on $X$ with a $|X|\geq 3$. Let $\widetilde N$ denote the topological LSA network $\mathcal N^\oplus$ modulo 2- and 3-cycles and directions of edges in 4-cycles. Then  for generic hybridization parameters under the mixture of coalescent mixtures model $\M$ on  $ \mathcal N^+$,  $\widetilde N$ is identifiable  from the theoretical logDet distances for pairs of taxa.
\end{theorem}
\begin{proof}
Propositions \ref{prop::2and3_1cyc}, \ref{prop::3_2cyc}, and \ref{prop::4cyc} imply that for generic parameters the three logDet distances for any choice of 3 taxa are distinct if, and only if, the induced rooted triple network has a 4-cycle. Moreover, the unrooted topology of the 4-cycle is determined by the largest of the three distances. Thus
the set $S$ of Proposition \ref{prop::combnet} is determined, yielding the result.
\qed\end{proof}

An example of a rooted level-1 network and the structure that we have shown to be identifiable from logDet distances under the model $\M$ is given in Figure
\ref{fig::logdetnet}. On the left is a level-1 rooted phylogenetic network
 with cycles of various sizes, and on the right the partially directed network that could be inferred from it for generic parameters.
 \begin{figure}
	\begin{center}
			\includegraphics{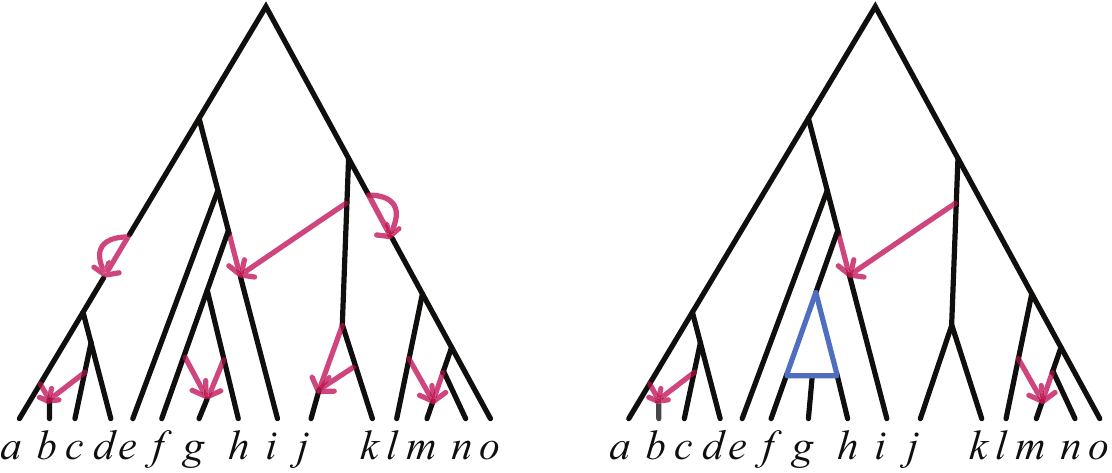} 
	\end{center}
\caption{(Left) A rooted binary level-1 network and (Right) that part of its structure that Theorem \ref{thm::main} identifies from logDet distances under the model $\M$  for generic parameters. Both $2$- and $3$- cycles are lost, as are the directions of $4$-cycle edges, and hence knowledge of the hybrid nodes in 4-cycles. Directed edges in cycles of size greater than 4 are identifiable.}\label{fig::logdetnet} 
\end{figure}

\section{Normalized triples of logDet distances.}\label{sec::other}

In the previous section, we obtained linear equalities and inequalities that the logDet distances between three taxa must satisfy if they are related by various level-1 rooted networks. Combined with the combinatorial result of Section \ref{sec::comb} these are sufficient for proving the identifiability claim that is the main focus of this work. However, it is worthwhile to seek a more complete characterization of what distances are achievable by various network topologies. In particular, with an eye toward practical application, any tighter characterizations would enable stronger testing for network topology from the empirical distances.

Here we conduct a partial investigation, characterizing not the triple of theoretical logDet distances that may be produced on rooted 3-taxon networks, but rather the \emph{normalized triple} obtained by dividing the distances by their sum. The triple of distances forms a point in the non-negative octant $\left (\mathbb R_{\ge0}\right )^3$, while the normalized triple gives a point in the 2-dimensional simplex. Thus plots can be made with the normalized distances that are analogous to the simplex plots for visualizing gene quartet concordance factors \cite{Banos2019,MAR2019,AMR2020}. 
Just as simplex plots of concordance factors aid in understanding genomic data sets, we anticipate that the 2-simplex visualization of the normalized logDet distance triples will be similarly useful.
 
We begin with the logDet triples from 3-taxon trees. 
 
\begin{proposition}
Let $\ell=(\ell_{ab},\ell_{ac},\ell_{bc})$ with $0<\ell_{ab}\le \ell_{ac}=\ell_{bc}$ be a triple of positive numbers summing to 1. Then there  exists an ultrametric rooted tree with topology $((a,b),c)$ and GTR substitution model parameters such that
the normalized theoretical logDet distances of sequences generated under the coalescent mixture model are $\ell$.
\end{proposition}

\begin{proof} Consider the metric species tree $((a\tc 0,b\tc 0)\tc x/2,c\tc x/2)$, and constant population sizes $\epsilon>0$ on all edges. Fix a single substitution model, say the Jukes-Cantor, for sequence generation. Since small population sizes $\epsilon$ result in rapid coalescence with arbitrarily high probability, by taking $\epsilon$ sufficiently small one can show the expected frequency array can be made arbitrarily close to that which would arise if all gene trees exactly matched the species tree.
Thus the theoretical logDet distances can be made arbitrarily close to $d_{LD}(a,b)=0$  and $d_{LD}(a,c)=d_{LD}(b,c)=x$, which normalizes to $(0,1/2,1/2)$.

The unresolved species tree $(a\tc x/2,b\tc x/2,c\tc x/2)$, regardless of choice of population functions on the edges yields, by exchangeability of the taxa, a triple of equal logDet distances, which normalizes to $(1/3,1/3,1/3)$.

While the two trees above have 0-length edges and hence are non-binary, perturbations to binary trees with positive length edges can produce normalized logDet distances that are arbitrarily close.

Since the normalized logDet distances are continuous functions of parameters, the parameter space is connected, and the image of the normalized distances lies in a line segment by Proposition \ref{prop::2and3_1cyc}, the claim follows.
\qed\end{proof}

We turn now to networks with a single cycle.

\begin{proposition}\label{prop:4cyclefills}
Let $\ell=(\ell_{ab},\ell_{ac},\ell_{bc})$ with $0<\ell_{ab}\le \ell_{ac}<\ell_{bc}$ be a triple of positive numbers summing to 1. Then there  exists a binary ultrametric rooted network on taxa $a,b,c$  with a single 4-cycle and GTR substitution model parameters such that
the normalized theoretical logDet distances of sequences generated under a single-class coalescent mixture model are $\ell$.
\end{proposition}

\begin{proof}
 The 4-cycle network we construct is shown in Figure \ref{fig::4c_bac}, with $t_0,t_1$ measured in generations, and the hybrid edges of length 0. Consider a single constant population size $N>0$ for all populations over the tree and above the root, and a Jukes-Cantor substitution process with constant rate $\mu>0$. We will choose values for $t_0, t_1>0$, $\gamma\in[1/2,1)$ so that the normalized distances for the coalescent mixture model with this single substitution process are given by $\ell$. 
 
 \begin{figure}
	\begin{center}
	\includegraphics{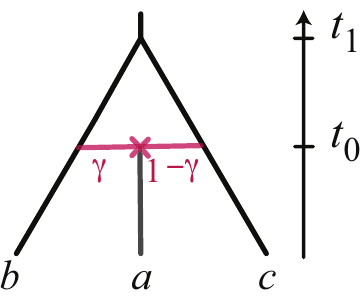}
	\end{center}
\caption{The 4-cycle network, with times in generations, constructed in Proposition \ref{prop:4cyclefills}. Hybridization parameters are $\gamma$, $1-\gamma$, and hybrid edges have length 0. }\label{fig::4c_bac}
\end{figure}

Recall that if $M(t)$ denotes the Jukes-Cantor Markov matrix for a substitution process over time $t$ with rate 1, then the common value of all
its off-diagonal entries is
$$f(t)=\frac 14\left (1- e^{-\frac 43 t}\right ) .$$ 
With $D=\diag (1/4,1/4,1/4,1/4)$, the Jukes-Cantor pattern frequency array is
$DM(t)$, and the logDet distance (equal to Jukes-Cantor distance) is $$t=f^{-1}(f(t))=-\frac 34 \log \left (1-4f(t) \right ).$$ Note that $f$ is an increasing function.

From equation 4.1 of  \cite{Allman2019}, for a coalescent mixture Jukes-Cantor model on an ultrametric  tree with uniform population size $N$ and mutation rate $\mu$, sequences for two taxa $x,y$ whose MRCA is at time $t$ before the present has expected pattern frequency array
$$F(t)=D M(2t\mu)\tilde M(\mu,N),$$
where $\tilde M(\mu,N)$ is a Markov matrix of  Jukes-Cantor form describing the expected additional substitutions due to the coalescent model delaying lineages merging until some time above the MRCA.
The logDet distance between $x,y$ is then the same as the Jukes-Cantor distance, which is computed to be
$$d_{LD}(x,y)=2t\mu + \beta$$
where $\beta=\beta(\mu,N)>0$ can be explicitly computed from $\tilde M(\mu, N)$, though we will not do so here. Since $\beta$ is continuous and $\beta(\mu,N)\to 0$ as $N\to 0$ and $\beta(\mu,N)\to \infty$ as $N\to \infty$,
it follows that $\beta$ takes on all positive values. 

Now by Lemma \ref{lem::3_1and4cycles} on the 4-cycle network of Figure \ref{fig::4c_bac} the expected pattern frequency array for $a,b$ is
$$\gamma F(t_0)+(1-\gamma) F(t_1) = D M_{ab} \tilde M(\mu, N)$$
where $$M_{ab}=\gamma M(2t_0\mu)+(1-\gamma) M(2t_1\mu)$$
has the usual Jukes-Cantor form, with 
off-diagonal entries 
$$f_{ab}=\gamma f(2t_0\mu)+(1-\gamma)f(2t_1\mu).$$
This shows
$$d_{LD}(a,b)=f^{-1}(f_{ab}) +\beta.$$

A similar calculation shows
$$d_{LD}(a,c)=f^{-1}(f_{ac}) +\beta,$$
where $$f_{ac}= \gamma f(2t_1\mu)+(1-\gamma)f(2t_0\mu).$$

The expected pattern frequencies for $b,c$ sequences is $F(t_1)$, so 
$$d_{LD}(b,c)=f^{-1}(f_{bc}) +\beta$$
where 
$$f_{bc} =f(2t_1\mu).$$

We now determine parameters which produce the normalized triple of distances $\ell$.
Fixing values of  $\mu$, $N$ determines a fixed value of $\beta>0$. Next, choose some value $m$ so that $$f\left (m\ell_{ab}-\beta\right)>\frac 1{8},$$ which can be done since
$f:\mathbb R^{>0} \to (0,1/4)$ is surjective and increasing. Then, with $x_{ij}=f \left (m\ell_{ij}-\beta\right)$, because $\ell_{ab}\le \ell_{ac}< \ell_{bc}$ we have
$$\frac 18<x_{ab}\le x_{ac}< x_{bc}<\frac 14.$$
Let $x_0=x_{ab}+x_{ac}-x_{bc}$, so $0< x_0<\frac 14$. Determine $t_0$ by $f(2t_0\mu)=x_0$, and $\gamma\in [1/2,1)$ by
$$\gamma=\frac {x_{bc}-x_{ab}}{2x_{bc}-x_{ab}-x_{ac}}, \text{ so } 1-\gamma= \frac {x_{bc}-x_{ac}}{2x_{bc}-x_{ab}-x_{ac}}.$$
Then choose $t_1$ by $f(2t_1\mu)=x_{bc}.$

To verify that these parameter choices give the desired normalized triple of distances,
the expected distance between $a,b$ is
\begin{align*}
d_{LD}(a,b)&=f^{-1}( \gamma f(2t_0\mu) +(1-\gamma)f(2t_1\mu ) ) +\beta\\
&=f^{-1}( \gamma x_0 +(1-\gamma)x_{bc}) +\beta\\
&=f^{-1}(x_{ab})+\beta\\
&=m\ell_{ab}.
\end{align*}
Similarly, we see $d_{LD}(a,c)=m\ell_{ac}$.
Finally we have
$$d_{LD}(b,c)=f^{-1}(f(2t_1\mu))+\beta= f^{-1}(x_{bc})+\beta=m\ell_{bc}.$$
\qed\end{proof}

Note that even if $\ell_{ac}=\ell_{bc}$, the argument of Proposition  \ref{prop:4cyclefills} can be modified slightly by taking $\gamma=1$ in the analytic continuation of the parameterization. However, that choice of the hybridization parameter essentially means that in place of a 4-cycle network parameter we have a tree.

\begin{figure}
	\begin{center}
	\includegraphics{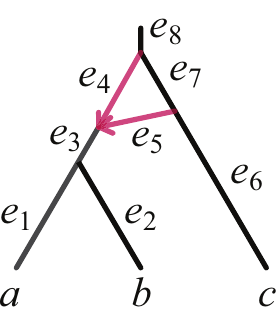}
	\end{center}
	\caption{A $3_2$-network, with numbered edges, as used in Proposition \ref{prop:32cyclefills}. The hybridization parameter on edge $e_5$ is $\gamma$, and on $e_4$ is $1-\gamma$.}\label{fig::32cycle}
\end{figure}

Finally, we consider a network with a $3_2$-cycle. While Proposition \ref{prop::3_2cyc} shows the  normalized triples of theoretical logDet distances
  lie on the same line as those for  a tree, we establish they need not be restricted to the same line segment of tree-like distances. However, we do not completely characterize the extent of the segment they fill out.

\begin{proposition}\label{prop:32cyclefills}
Let $\ell=(\ell_{ab},\ell_{ac},\ell_{bc})$ with $\ell_{ac}= \ell_{bc}$ be a triple of positive numbers summing to 1 with $0<\ell_{ab}<\frac 12$. Then there exists a binary ultrametric rooted network on taxa $\{a,b,c\}$ with a single $3_2$-cycle  whose leaf-descendants are $a,b$ and GTR substitution model parameters such that
the normalized theoretical LogDet distances of sequences generated under the coalescent mixture model are $\ell$.
\end{proposition}

\begin{proof}
We construct several $3_2$-cycle species networks of the form shown in Figure \ref{fig::32cycle}, with edge lengths $t_i=\ell(e_i)$.
In making choices of numerical parameters, since the network is ultrametric we view $t_1,t_3,t_5,t_7$ as independent, determining $t_2,t_4,t_6$. 
The population size on edge $e_i$ for $3\le i \le 8$ are constants $N_i$, with the sizes on terminal edges irrelevant. The hybridization parameters are $1-\gamma$ and $\gamma$ on edges  $e_4$ and $e_5$ respectively.  We also fix a single Jukes-Cantor substitution process with any constant rate $\mu>0$.  

By Proposition \ref{prop::3_2cyc}, for any choices of the $t_i,N_i,\gamma$, the theoretical LogDet distances will satisfy $d_{LD}(a,c)=d_{LD}(b,c)$ so the normalized theoretical LogDet distance triple lies on a line.
Since the parameter space is connected, it is  enough to show that 
\begin{equation}
\frac{d_{LD}(a,b)}{2d_{LD}(a,c)+d_{LD}(a,b)}\label{eq:ndist}
\end{equation} 
is arbitrarily close to $0$ for some choice of the parameters, and arbitrarily close to $1/2$ for others, to conclude that the rescaled expected distances give all the described triples.

To make  expression \eqref{eq:ndist} near $0$, we choose parameters  with $t_1$ and $N_3$ sufficiently small so that with high probability the $a,b$ lineages coalesce quickly.
Specifically, let $t_3=1$, and fix any positive values for $t_5,t_7$ and $N_i$ for $i\ne 3$. Now for any $\epsilon>0$, as $N_3\to 0^+$, the probability of lineages from $a,b$ coalescing on  $e_3$ within $\epsilon$ of entering it approaches 1.
Using this, it is straightforward to show that as $N_3\to 0^+$ the expected pattern frequency array for $a,b$ approaches that for the JC model on a 2-taxon tree of total length $2t_1$. This then implies that $d_{LD}(a,b)\to 2\mu t_1$ as $N_3\to 0^+$. On the other hand, for all values of $N_3>0$ one can show $d_{LD}(a,c)>2\mu (t_1+2)$. Thus for a sufficiently small choices of $t_1$ and $N_3$, we can make $d_{LD}(a,b)/(2d(a,c)+d(a,b))$ as close to $0$ as desired.

To produce a value of  expression \eqref{eq:ndist} near $1/2$ is more subtle. We choose parameters so that $a,b$ lineages are likely to enter $e_5$, but if they both do they are then unlikely to coalesce in it, and coalescence of any pair of lineages in $e_7$ is likely to occur quickly. First set $t_5=0$, $t_7=1$ and $N_8$ arbitrary.
For any $t_1,t_3$ and $\gamma$,  by choosing $N_3=N_4=N_5$ sufficiently large, the probability that
the $a,b$ lineages coalesce on $e_3$, $e_4$, or $e_5$ can be made arbitrarily small, so that if they coalesce below the root with (conditional) probability approaching 1 they must do so on $e_7$. This requires that both the $a,b$ lineages follow $e_5$, which occurs with probability $\gamma^2$. If lineages $a,c$ coalesce below the root, they must do so on $e_7$, requiring the $a$ lineage to follow $e_5$, which occurs with probability $\gamma$. 
By picking $N_7$ sufficiently small, the probability that two lineages in edge $e_7$ coalesce near the lower end can be made close to 1. All this shows that once $t_1,t_3$ and $\gamma$ are chosen, by appropriate choices of the $N_i$ we can ensure the expected frequency arrays for $a,b$
and $a,c$ are arbitrarily close to
$$ \gamma^2 F(t_1+t_3)  +  (1-\gamma^2) G(t_1+t_3+1,N_8)$$
and
$$ \gamma F(t_1+t_3)  +  (1-\gamma) G(t_1+t_3+1,N_8),$$
respectively, where $F(t)$ is the expected pattern frequency array for two samples at distance $2t$ and $G(t,N)$ is the expected array under the coalescent
for 2 lineages which enter a common population of size $N$ at time $t$. Further picking sufficiently small values for $t_1,t_3$, the pattern frequency arrays for $a,b$ and $a,c$  can be made arbitrarily close to
$$ \gamma^2\frac 14 I  +  (1-\gamma^2) G(1,N_8)
$$ 
and
$$ \gamma \frac 14 I  +  (1-\gamma) G(1,N_8),
$$ 
respectively. Thus for any $\gamma$ the theoretical distance can be made arbitrarily close to the distance computed from the above arrays.
Using the formulas defined in the proof of Proposition \ref{prop:4cyclefills}, we find these distances are 
$$d_{LD}(a,b)=f^{-1}\left ( (1-\gamma^2) \delta \right)$$
and
$$d_{LD}(a,c)=f^{-1}\left ( (1-\gamma) \delta \right)$$
where $\delta>0$ is the off-diagonal entry of $G(1,N_8)$.
Thus once $\gamma$ is specified, by choosing $t_1$, $t_3$, $N_3=N_4=N_5$, $N_7$ we can ensure  expression \eqref{eq:ndist}
is arbitrarily close to
\begin{equation}
\frac {\log(1-4\delta(1-\gamma^2))}{2\log(1-4\delta(1-\gamma))+\log(1-4\delta(1-\gamma^2))}.
\label{eq:gamma}
\end{equation}
Applying L'Hopital's rule shows the limit of expression \eqref{eq:gamma} as $\gamma\to 1$ is $\frac 12$. Thus for any $\epsilon>0$, by first choosing $\gamma$ near 1 so that the  expression \eqref{eq:gamma} is within $\epsilon/2$ of $1/2$, and then choosing $t_1=t_3$, $N_3=N_4=N_5$, $N_7$ so that  expression \eqref{eq:ndist} is within $\epsilon/2$ of expression \eqref{eq:gamma},
we obtain the desired result.
\qed\end{proof}

The results of this section, combined with those of Section \ref{sec::logDet} are summarized by Figure \ref{fig::regions}, which indicates the various regions of the simplex which normalized logDet triples fill, according to whether the network has a $4$-cycle, a $3_2$-cycle, or neither.
 \begin{figure}
 	\begin{center} 
	\includegraphics{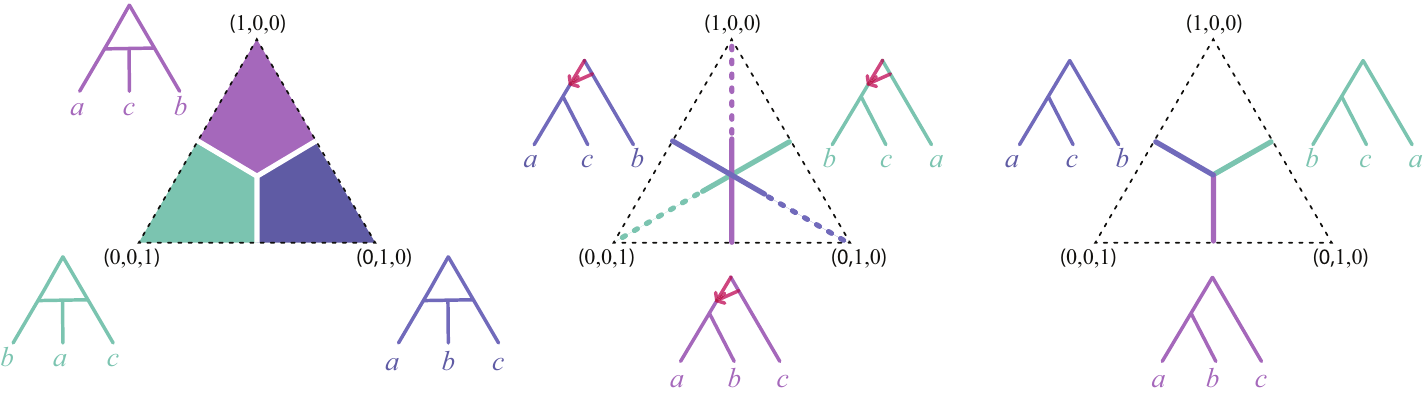}
	\end{center}
	\caption{The regions of the simplex filled by normalized triples of logDet distances under the model $\mathcal M$ on a 3-taxon network. The networks shown are those obtained by suppressing all cycles other than $4$- and $3_2$-cycles, and then undirecting the 4-cycle edges. Normalized logDet distances are ordered as $(\ell_{ab},\ell_{ac},\ell_{bc}).$ Networks with $3_2$-cycles fill the solid line segments in the center simplex, but it is unknown whether they may also produce points in the dashed line segments.}\label{fig::regions}
\end{figure}

\section{Conclusion}\label{sec:discussion}

Theorem \ref{thm::main} states that most topological features of an ultrametric level-1 network can be identified from theoretical logDet distances under a fairly general model of sequence evolution with incomplete lineage sorting.  It more generally implies network identifiability from pattern frequency arrays, since logDet distances are functions of these.
In particular, individual gene trees, or even sequences partitioned into genes, are not required for network identifiability.

While identifiability is a theoretical question about the model, it has important implications for data analysis. Indeed, it is a key requirement for a statistically consistent inference procedure to exist. While our method of proof of identifiability, using the logDet distance, suggests using that distance as a basis for an inference procedure, others might be developed as well. 

In subsequent work, we will explore using the logDet distance in a procedure for level-1 network inference following the framework of NANUQ \cite{ABR2019}. In outline, for each triple of taxa, the location of the normalized triple of logDet distances in simplex plots such as those of Figure \ref{fig::regions} can indicate whether the rooted triple has a 4-cycle or not. A triple near the lines through the centroid can, through some statistical test, be judged unlikely to have arisen from a 4-cycle, while those farther away are judged to have arisen from a 4-cycle. Then, modifying the rooted triple distance of \cite{Rhodes2019} to a network setting, similarly to how NANUQ modified the quartet distance, an intertaxon distance can be computed from the results of these statistical tests. Rules for relating a splits graph for the expected rooted triple distance to the original network will be developed. When applied to the splits graph constructed by NeighborNet from the empirically-derived distance, this should lead to consistent network inference. Since individual gene trees are never inferred, this will potentially give a much faster data analysis pipeline than the current version of NANUQ, which is built on quartet concordance factors across gene trees.

\section{Acknowledgements}

This work was supported by the National Institutes of Health [R01 GM117590],
under the Joint DMS/NIGMS Initiative to Support Research at the Interface of the Biological  Mathematical Sciences, and
[2P20GM103395], an NIGMS Institutional Development Award (IDeA).

\bibliographystyle{alpha}
\bibliography{Hybridization}

\end{document}